\documentclass[11pt, english]{article}
\pdfoutput=1
\usepackage{babel}
\usepackage{amsmath, amssymb, amsthm, eucal}
\usepackage{graphicx}
\usepackage[shortlabels]{enumitem}

\usepackage[textheight=630pt,
  textwidth=468pt,
  centering]{geometry}

\newtheorem{theorem}{Theorem}

\newtheorem{lemma}{Lemma}
\newtheorem{proposition}{Proposition}

\theoremstyle{definition}
\newtheorem{definition}{Definition}

\newtheorem{remark}{Remark}

\newcommand{\eps}{\varepsilon}

\newcommand{\abs}[1]{\left| #1 \right|}

\newcommand{\norm}[1]{\Vert #1 \Vert}

\newcommand{\RN}{\mathbb{R}^n}

\newcommand{\RR}{\mathbb{R}}

\newcommand{\PP}{\mathbb{P}}

\newcommand{\brac}[1]{\left(#1\right)}
\newcommand{\fbrac}[1]{\left\{#1\right\}}
\newcommand{\sbrac}[1]{\left[#1\right]}
\newcommand{\abrac}[1]{\langle#1\rangle}

\DeclareMathOperator{\conv}{conv}

\DeclareMathOperator{\mean}{\mathbb{E}}
\DeclareMathOperator{\supp}{supp}

\DeclareMathOperator{\sgn}{sgn}
\DeclareMathOperator{\cone}{cone}
\DeclareMathOperator{\Id}{\mathbf{Id}}

\begin{document}

\title{Analysis $\ell_1$-recovery with frames and Gaussian measurements}
\author{Holger Rauhut\thanks{RWTH Aachen University, Lehrstuhl C f{\"u}r Mathematik (Analysis), Templergraben 55, 52062 Aachen Germany, \tt{rauhut@mathc.rwth-aachen.de}}, Maryia Kabanava\thanks{RWTH Aachen University, Lehrstuhl C f{\"u}r Mathematik (Analysis), Templergraben 55, 52062 Aachen Germany, \tt{rauhut@mathc.rwth-aachen.de}} }
\date{\today}

\maketitle

%%%%%%%%%%%%%%%%%%%%%%%%%%%%%%%%%%%%%%%%%%%%%%%%%%

\begin{abstract}
This paper provides novel results for the recovery of signals from undersampled measurements based on analysis $\ell_1$-minimization, 
when the analysis operator is given by a frame. We both provide so-called uniform and nonuniform recovery guarantees for cosparse (analysis-sparse) signals using
Gaussian random measurement matrices.
The nonuniform result relies on a recovery condition via tangent cones and the uniform recovery guarantee is based on an analysis version of the null space property. 
Examining these conditions for Gaussian random matrices leads to precise bounds on the number of measurements required for successful recovery. 
In the special case of standard sparsity, our result improves a bound due to Rudelson and Vershynin concerning the exact reconstruction of sparse signals from Gaussian measurements with respect to the constant and extends it to stability under passing to approximately
sparse signals and to robustness under noise on the measurements.
\end{abstract}

\medskip
{\bf Keywords:} compressive sensing, $\ell_1$-minimization, analysis regularization, frames, Gaussian random matrices.

%%%%%%%%%%%%%%%%%%%%%%%%%%%%%%%%%%%%%%%%%%%%%%%%%%%
\section{Introduction}

Compressive sensing \cite{do06-2,carota06,FoucartRauhut} is a recent field that has seen enormous research activity in the past years. It predicts that certain signals (vectors) can be recovered from 
what was previously believed to be incomplete information using efficient reconstruction methods. Applications of this principle range from magnetic resonance imaging over radar and remote sensing to astronomical
signal processing and more. The key assumption of the theory is that the signal to be recovered is sparse 
or can at least be well-approximated by a sparse one. Most research activity so far has been dedicated to the synthesis sparsity model where one assumes that the signal can be written as a 
linear combination of only a small number of elements from a basis, or more generally an overcomplete frame. In certain situations, however, it turns out to be more efficient to work with an analysis-based sparsity model.
Here, one rather assumes that the application of a linear map yields a vector with a large number of zero entries. While the synthesis and the analysis model are equivalent in special cases, 
%the analysis model is more powerful in general. 
they are very distinct in an overcomplete case. By now, comparably few investigations have been dedicated to the analysis sparsity model and its rigorous understanding is still in its infancy.

The analysis based sparsity model and corresponding reconstruction methods were introduced systematically in recent work of Nam et al.~\cite{NamDaviesEladGribonval}.
Nevertheless we note that it appeared also in earlier works, see e.g.~\cite{CandesEldarNeedellRandall}. In particular, the popular method 
of total variation minimization \cite{ChanShen,NeedellWard} in image processing is closely related to analysis based sparsity with respect to a difference operator. An estimate of the number of Gaussian measurements for successful recovery via total variation minimization has been recently obtained in \cite{KabanavaRauhutZhang}.

In this paper we consider the analysis based sparsity model for the important case that the analysis transform is given by inner products with respect to a possibly redundant frame. 
As reconstruction method we study a corresponding analysis $\ell_1$-minimization approach. Furthermore, we assume that the linear measurements are obtained via an application
of a Gaussian random matrix. The main results of this paper provide precise estimates of the number of measurements required for
the reconstruction of a signal whose analysis representation has a given number of zero elements. Moreover, stability estimates are given. An alternative bound on the number of measurements can be found in \cite{KabanavaRauhutZhang}.

\subsection{Problem statement and main results}

We consider the task of reconstructing a signal $\mathbf{x}\in\RR^d$ from incomplete and possibly corrupted measurements given by
\begin{equation}\label{eqNoisyMeasurements}
\mathbf{y}=\mathbf{M}\mathbf{x}+\mathbf{w},
\end{equation}
where $\mathbf{M}\in\RR^{m\times d}$ with $m\ll d$ is the measurement matrix and 
$\mathbf{w}$ corresponds to noise. Since this system is underdetermined it is impossible to recover $\mathbf{x}$ from $\mathbf{y}$ without additional information, even when $\mathbf{w} = \mathbf{0}$. 

As already mentioned, the underlying assumption in compressive sensing is sparsity. The \emph{synthesis sparsity prior} assumes that $\mathbf{x}$ 
can be represented as a linear combination of a small number of elements of a dictionary $\mathbf{D}\in\RR^{d\times n}$, i.e.,
\[
\mathbf{x}=\mathbf{D\alpha},\;\;\mathbf{\alpha}\in\RN,
\]
where the number of non-zero elements of $\mathbf{\alpha}$, denoted by $\norm{\alpha}_0$, is considerably less than $n$. 
Often $\mathbf{D}$ is chosen as a unitary matrix, which refers to sparsity of $\mathbf{x}$ in an orthonormal basis. 
Unfortunately, the approach to recover $\alpha$, or $\mathbf{x}$ respectively, from $\mathbf{y} = \mathbf{M} \mathbf{x} = \mathbf{M D \alpha}$ (assuming the noiseless case for simplicity) 
via $\ell_0$-minimization, i.e.,
\[
\underset{\alpha\in\RR^n}\min\; \|\mathbf{\alpha}\|_0 \quad \mbox{ subject to } \mathbf{M D \alpha} = \mathbf{y},
\]
is NP-hard in general. A by-now well-studied tractable alternative is the $\ell_1$-minimization approach of finding the minimizer $\alpha^*$ of
\begin{equation}\label{eqBP}
\underset{\mathbf{\alpha}\in\RR^n}\min \; \|\mathbf{\alpha}\|_1 \quad \mbox{ subject to }  \mathbf{M D \alpha} = \mathbf{y}
\end{equation}
The restored signal is then given by $\mathbf{x}^* = \mathbf{D\alpha^*}$. This optimization problem is referred to as basis pursuit \cite{ChenDonohoSaunders}. In the noisy case, one passes to
\begin{equation}\label{eqBPDN}
\underset{\mathbf{\alpha}\in\RR^n}\min \; \|\mathbf{\alpha}\|_1 \quad \mbox{ subject to } \| \mathbf{M D \alpha} - \mathbf{y} \|_2 \leq \eta,
\end{equation}
where $\eta$ corresponds to an estimate of the noise level. 

The \emph{analysis sparsity prior} assumes that $\mathbf{x}$ is sparse in some transform domain, that is, 
given a matrix $\mathbf{\Omega}\in\RR^{p\times d}$ -- the so-called {\em analysis operator} --  the vector $\mathbf{\mathbf{\Omega x}}$ is sparse. 
For instance, such operators can be generated by the discrete Fourier transform, the finite difference operator (related to total variation), 
wavelet \cite{Mallat,RonShen,SelesnickFigueiredo}, curvelet \cite{CandesDonoho} or Gabor transforms \cite{gr01}. 

Analogously to \eqref{eqBP}, a possible strategy for the reconstruction of analysis-sparse vectors (or cosparse vectors, see below) is to solve the analysis $\ell_1$-minimization problem
\begin{equation}\label{eqProblemP1}
\underset{\mathbf{z}\in\RR^d}\min\;\norm{\mathbf{\Omega z}}_1\;\;\mbox{subject to}\;\;\; \mathbf{M}\mathbf{z}=\mathbf{y}, 
\end{equation}
or, in the noisy case,
\begin{equation}\label{eqProblemP1Noise}
\underset{\mathbf{z}\in\RR^d}\min\;\norm{\mathbf{\Omega z}}_1\;\;\mbox{subject to}\;\;\norm{\mathbf{M}\mathbf{z}-\mathbf{y}}_2\leq\eta.
\end{equation}
Both optimization problems can be solved efficiently using convex optimization techniques, see e.g.~\cite{BoydVandenberghe}.
If $\mathbf\Omega$ is an invertible matrix, then these analysis $\ell_1$-minimization problems are equivalent
to \eqref{eqBP} and \eqref{eqBPDN}. However, in general the analysis $\ell_1$-minimization problems cannot be reduced to the standard $\ell_1$-minimization problems.

We note, that one may also pursue greedy or other iterative approaches for recovery, see e.g.~\cite{FoucartRauhut} for an overview in the standard synthesis sparsity case and see e.g.~\cite{NamDaviesEladGribonval} for the analysis sparsity case. However, we will concentrate on the above optimization approaches here.

In the remainder of this paper, we assume that the analysis operator is given by a frame. 
Put formally, let $\{\mathbf{\mathbf{\mathbf{\omega}}}_i\}_{i=1}^p$, $\mathbf{\omega}_i\in\RR^d$, $p\geq d$, be a frame, i.e., there exist positive constants $A$, $B>0$ 
such that for all $\mathbf{x}\in\RR^d$
\[
A\norm{\mathbf{x}}_2^2\leq\sum\limits_{i=1}^p\abs{\abrac{\mathbf{\omega}_i,\mathbf{x}}}^2\leq B\norm{\mathbf{x}}_2^2.
\]
Its elements are collected as rows of the matrix $\mathbf{\Omega}\in\RR^{p\times d}$. The analysis representation of a signal $\mathbf{x}$ is given by 
the vector $\mathbf{\mathbf{\Omega x}}=\fbrac{\abrac{\mathbf{\omega}_i,\mathbf{x}}}_{i=1}^p\in\RR^p$. 
(We note that in the literature it is often common to collect the elements of a frame rather as columns of a matrix. However, for our purposes it is more convenient to collect them as rows.)
The frame is called tight if the frame bounds coincide, i.e., $A = B$.

Cosparsity is now defined as follows.
\begin{definition}\label{defCosparsity}
Given an analysis operator $\mathbf{\Omega}\in\RR^{p\times d}$, the cosparsity of $\mathbf{x}\in\RR^d$ is defined as
\[
l:= p - \norm{\mathbf{\Omega x}}_0.
\]
The index set of the zero entries of $\mathbf{\Omega x}$ is called the cosupport of $\mathbf{x}$. If $\mathbf{x}$ is $l$-cosparse, then $\mathbf{\Omega x}$ is $s$-sparse with $s = p - l$.
\end{definition}
The motivation to work with the cosupport rather than the support in the context of analysis sparsity is that in contrast to synthesis sparsity, 
it is the location of the {\it zero}-elements which define a corresponding subspace. In fact, if $\Lambda$ is the cosupport of $\mathbf{x}$, then it follows from Definition~\ref{defCosparsity} that
\[
\abrac{\mathbf{\omega}_j,\mathbf{x}}=0,\quad \mbox{ for all } j\in\Lambda.
\]
Hence, the set of $l$-cosparse signals can be written as 
\[
\bigcup_{\Lambda \subset [p]: \# \Lambda = l} W_\Lambda,
\]
where $W_\Lambda$ denotes the orthogonal complement of the linear span of $\{\mathbf{\omega}_j : j \in \Lambda\}$. 

In contrast to standard sparsity, there are often certain restrictions on the values that the cosparsity can take. In fact, in the generic case that the frame elements $\mathbf{\omega}_j$ are in general position
in $\mathbb{R}^d$, then every set of $d$ rows of $\mathbf{\Omega}$ are linearly independent. Then the maximal number of zeros that can be achieved for a nontrivial vector $\mathbf{x}$ in the analysis representation $\mathbf{\Omega x}$ is less than $d$, since otherwise $\mathbf{x}=0$. 
Thus, for the cosparsity $l$ of any non-zero vector $\mathbf{x}$ it holds $l< d$ in this case. Nevertheless, if there are linear dependencies among the frame elements 
$\mathbf{\omega}_j$, then larger values of $l$ are possible. This applies to certain redundant frames as well as to the difference operator (related to total variation). 

Our main results concern the minimal number $m$ of measurements that are necessary to recover an $l$-cosparse vector $\mathbf{x}$ from
$\mathbf{y} = \mathbf{M x}$ with $\mathbf{M} \in \mathbb{R}^{m \times d}$. As it is hard to come up with theoretical guarantees for deterministic matrices $\mathbf{M}$, we pass to
random matrices. As common in compressive sensing, we work with Gaussian random matrices, that is, with matrices having independent standard normal distributed entries.
Gaussian random matrices have already proven to provide accurate theoretical guarantees in the context of standard synthesis sparsity,
see e.g.~\cite{ChandrasekaranRechtParriloWillsky,dota09}. Moreover, empirical tests indicate that also other types of random matrices behave very similar to Gaussian random matrices in terms 
of recovery performance \cite{dota09-1}, although a theoretical justification may be much harder than for Gaussian matrices. 

We both provide so-called nonuniform and uniform recovery guarantees. The nonuniform result states that a given fixed cosparse vector $\mathbf{x}$ is recovered via analysis $\ell_1$-minimization
from $\mathbf{y} = \mathbf{M x}$ with high probability using a random choice of a Gaussian measurement matrix $\mathbf{M}$ under a suitable condition on the number of measurements. 
In contrast, the uniform result states that a 
single random draw of a Gaussian matrix $\mathbf{M}$ is able to recover {\em all} cosparse signals $\mathbf{x}$ simultaneously with high probability. Clearly, the uniform statement 
is stronger than the nonuniform one, however, as we will see, the uniform statement requires more measurements.

We start with the nonuniform guarantee for recovery of cosparse signals with respect to frames using Gaussian measurement matrices.
\begin{theorem}\label{thMainResultForFrame}
%Let $\Omega \mathbf{x}\in\RR^p$ be an $s$-sparse vector, 
Let $\mathbf{\Omega} \in \mathbb{R}^{p \times d}$ be a frame with frame bounds $A,B>0$ and
let $\mathbf{x}$ be $l$-cosparse, that is, $\mathbf{\Omega x}$ is $s$-sparse with $s = p-l$. For a Gaussian random matrix $\mathbf{M}\in\RR^{m\times d}$  
and $0<\eps<1$, if
\begin{equation}\label{eqNumberOfMeasurementsForFrame}
\frac{m^2}{m+1}\geq \frac{2Bs}{A}\brac{\sqrt{\ln\left(\frac{ep}{s}\right)}+\sqrt{\frac{A\ln(\eps^{-1})}{Bs}}}^2,
\end{equation}
then with probability at least $1-\eps$, the vector $\mathbf{x}$ is the unique minimizer of $\norm{\mathbf{\Omega z}}_1$ subject to $\mathbf{Mz}=\mathbf{Mx}$.
\end{theorem}
Roughly speaking, that is, ignoring terms of lower order, a fixed $l$-cosparse vector is recovered with high probability from 
\[
m > 2(B/A) s \ln(ep/s)
\] 
Gaussian measurements where $s = p-l$. Note that the number of measurements increases with increasing frame ratio $B/A$, and the optimal behavior occurs for tight frames.
For $\mathbf{\Omega} =\Id$, this bound slightly strengthens the main result for sparse recovery in \cite{ChandrasekaranRechtParriloWillsky}. We will also show stability of the reconstruction with respect to noise on the measurements, see Theorem~\ref{thNoisyMeasurements} below.
The proof of the above result (given in Section~\ref{sec:nonuniform}) relies on a characterization of the minimizer via tangent cones 
(Theorems~\ref{thRecoveryViaTangentCones} and \ref{thRecoveryViaTangentConesWithNoise}) which is similar to corresponding conditions stated 
in \cite{ChandrasekaranRechtParriloWillsky,MendelsonPajorTomczakJaegermann}. Moreover, our proof uses an extension of the Gordon's escape through a mesh theorem 
(Theorem~\ref{thModifiedGordonsEscapeThroughTheMesh}).

\medskip

We now pass to the uniform recovery result which additionally takes into account that in practice the signals are often only approximately cosparse.
The quantity
\[
\sigma_{s}(\mathbf{\Omega x})_1:=\inf\fbrac{\norm{\mathbf{\Omega x}-\mathbf{u}}_1: \mathbf{u} \;\mbox{is $s$-sparse}}
\]  
describes the $\ell_1$-best approximation error to $\mathbf{\Omega x}$ by $s$-sparse vectors. 
\begin{theorem}\label{thUniformRecoveryWithFrame}
Let $\mathbf{M}\in\RR^{m\times d}$ be a Gaussian random matrix, $0<\rho<1$ and $0<\eps<1$. If
\begin{equation}\label{eqNumberOfMeasurementsForFrameUniformRecovery}
\frac{m^2}{m+1}\geq \frac{2Bs\brac{1+(1+\rho^{-1})^2}}{A}\brac{\!\!\sqrt{\ln\frac{ep}{s}}+\frac{1}{\sqrt 2}+\sqrt{\frac{A\ln(\eps^{-1})}{Bs\brac{1+(1+\rho^{-1})^2}}}}^{\!\!2},
\end{equation}
then with probability at least $1-\eps$ for every vector $\mathbf{x}\in\RR^d$ a minimizer $\mathbf{\hat x}$ of $\norm{\mathbf{\Omega z}}_1$ subject to $\mathbf{Mz}=\mathbf{Mx}$ approximates $\mathbf{x}$ with $\ell_2$-error
\[
\norm{\mathbf{x}-\mathbf{\hat{x}}}_2\leq\frac{2(1+\rho)^2}{\sqrt{A}(1-\rho)}\frac{\sigma_{s}(\mathbf{\Omega x})_1}{\sqrt{s}}.
\]
\end{theorem}
Roughly speaking, with high probability every $l$-cosparse vector can be recovered via analysis $\ell_1$-minimization using a single random draw of a Gaussian matrix 
if 
\begin{equation}\label{eqApprNumberOfMeasurementsUniformRecovery}
m > 10 (B/A)s \ln(ep/s).
\end{equation}
Moreover, the recovery is stable under passing to approximately cosparse vectors when adding slightly more measurements.
The proof of this theorem relies on an extension of the null space property, which is well known in the synthesis sparsity case \cite{codade09,FoucartRauhut,grni03} and was adapted to the analysis sparsity setting in \cite{AldroubiChenPowell,Foucart}.
In fact, for the standard case $\mathbf{\Omega} = \Id$, we improve a result of Rudelson and Vershynin \cite{RudelsonVershynin} (also relying on the null space property) 
with respect to the constants in \eqref{eqNumberOfMeasurementsForFrameUniformRecovery} and add stability in $\ell_2$. 
We further show that recovery is robust under perturbations of the measurements in Theorem \ref{thRobustUniformRecoveryWithFrame}. We note, that in the standard exact sparse case with no noise, the constant $8$ in (\ref{eqApprNumberOfMeasurementsUniformRecovery}) can be replaced by $2e$, see the contribution by Donoho and Tanner in \cite{dota09}. Their methods, however, are completely different to ours, and it is not clear, whether they can be extended to analysis sparsity.

\subsection{Related work}

Let us discuss briefly related theoretical studies on recovery of analysis sparse vectors and compare them with our main results.
An earlier version of Theorem~\ref{thUniformRecoveryWithFrame} was shown by Cand{\`e}s and Needell in \cite{CandesEldarNeedellRandall}. However, they were only able
to treat the case that the analysis operator is given by a tight frame, that is, when $A=B$. Moreover, their analysis is based on a version of
the restricted isometry property and does not provide explicit constants in the corresponding bound on the required number of measurements. To be fair, we note, however, that their analysis applies to
general subgaussian random matrices. The results of \cite{CandesEldarNeedellRandall} were extended to the case of non-tight frames and Weibull matrices in the work of Foucart in \cite{Foucart}. The analysis in \cite{Foucart} incorporates the robust null space property, the verification of which for the Weibull matrices relies on a variant of the classical restricted isometry property. In our work we prove that Gaussian random matrices satisfy the robust null space property by referring to a modification of the Gordon's escape through a mesh theorem.  

A recent contribution by Needell and Ward \cite{NeedellWard} provides theoretical recovery guarantees for the special case of total variation minimization, which corresponds
to analysis $\ell_1$-minimization with a certain difference operator. Unfortunately, we cannot cover this situation with our main results because the difference operator is not a frame. Nevertheless, it would be interesting to pursue theoretical recovery guarantees for total variation minimization and Gaussian random matrices using the approach of this paper.

Nam et al.'s work \cite{NamDaviesEladGribonval} provides a systematic introduction of the analysis sparsity model and treats also greedy recovery methods, see also \cite{ginaelgrda13}. Further contributions are contained in \cite{LiuMiLi,VaiterPeyreDossalFadili}.

Our nonuniform recovery guarantees rely on a geometric characterization of the successful recovery. We obtain quantitative estimates by bounding a certain Gaussian width which can be thought as an intrinsic complexity measure. The authors of \cite{AmelunxenLotzMcCoyTropp} exploit the geometry of optimality conditions to study phase transition phenomena in random linear inverse problems and random demixing problems. They express their results in terms of the statistical dimension which is essentially equivalent to the Gaussian width, see Section 10. 3 of \cite{AmelunxenLotzMcCoyTropp} for further details. 

Also, we note that the optimization problems (\ref{eqProblemP1}) and (\ref{eqProblemP1Noise})
often appear in image processing \cite{CaiOsherShen,ChanShen}.

\subsection{Notation}

We use the notation $\mathbf{\Omega}_{\Lambda}$ to refer to a submatrix of $\mathbf{\Omega}$ with the rows indexed by $\Lambda$; $\brac{\mathbf{\Omega x}}_S$ stands for the vector whose entries indexed by $S$ coincide with the entries of $\mathbf{\Omega x}$ and the rest are filled by zeros. As we have already mentioned, the $\ell_0$-norm $\|\cdot\|_0$ of a vector corresponds to the number of non-zero elements in it. The unit ball in $\RR^d$ with respect to the $\ell_q$-norm is denoted by $B_q^d$. The operator norm of a matrix $\mathbf{A}$ is defined by $\norm{\mathbf{A}}_{2\to2}:=\underset{\norm{x}_2\leq 1}\sup\norm{\mathbf{Ax}}_2$ and the Frobenius norm is given by 
\[
\norm{\mathbf{A}}_F:=\brac{\sum_{i=1}^m\sum_{j=1}^d \abs{A_{ij}}^2}^{1/2}.
\]
It is well-known that Frobenius norm dominates the operator norm, $\norm{\mathbf{A}}_{2\to2} \leq \norm{\mathbf{A}}_F$.
Finally, $[p]$ is the set of all natural numbers not exceeding $p$, i.e., $[p]=\{1,2,\ldots,p\}$.

%%%%%%%%%%%%%%%%%%%%%%%%%%%%%%%%%%%%%%%%%%%%%%%%%%%%
%%%%%%%%%%%%%%%%%%%%%%%%%%%%%%%%%%%%%%%%%%%%%%%%%%%%

\section{Nonuniform recovery}
\label{sec:nonuniform}

In this section we prove Theorem~\ref{thMainResultForFrame} and extend it to robust recovery in Theorem~\ref{thNoisyMeasurements}.
The proof strategy is similar as in \cite{ChandrasekaranRechtParriloWillsky}. We rely on conditions on the measurement matrix $\mathbf{M}$ involving tangent cones, which should be of independent interest. In order to check these conditions for a Gaussian random matrix we
rely on an extension of the Gordon's escape through a mesh theorem. (In contrast to \cite{ChandrasekaranRechtParriloWillsky}, the standard version
of the Gordon's result is not sufficient for our purposes.)

%%%%%%%%%%%%%%%%%%%%%%%%%%%%%%%%%%%%%%%%%%%%%%%%%%%%%%%

\subsection{Recovery via tangent cones}

Our conditions for successful recovery of cosparse signals are formulated via tangent cones. For fixed $\mathbf{x}\in\RR^d$ we define the convex cone
\[
T(\mathbf{x})=\cone\{\mathbf{z}-\mathbf{x}: \mathbf{z}\in\RR^d,\;\norm{\mathbf{\Omega z}}_1\leq\norm{\mathbf{\Omega x}}_1\},
\]
where the notation ``cone'' stands for the conic hull of the indicated set. The following result is analogous to Proposition~2.1 in \cite{ChandrasekaranRechtParriloWillsky}.
\begin{theorem}\label{thRecoveryViaTangentCones}
Let $\mathbf{M}\in\RR^{m\times d}$. A vector $\mathbf{x}\in\RR^d$ is the unique minimizer of $\norm{\mathbf{\Omega z}}_1$ subject to $\mathbf{M}\mathbf{z}=\mathbf{Mx}$ if and only if $\ker \mathbf{M}\cap T(\mathbf{x})=\{\mathbf{0}\}$.
\end{theorem}
\begin{proof}
First assume that $\ker \mathbf{M}\cap T(\mathbf{x})=\{\mathbf{0}\}$. Let $\mathbf{z}\in\RR^d$ be a vector that satisfies $\mathbf{M}\mathbf{z}=\mathbf{Mx}$ and $\norm{\mathbf{\Omega z}}_1\leq\norm{\mathbf{\Omega x}}_1$.
This means that $\mathbf{z}-\mathbf{x}\in T(\mathbf{x})$ and $\mathbf{z}-\mathbf{x}\in\ker \mathbf{M}$. According to our assumption we conclude that $\mathbf{z}-\mathbf{x}=\mathbf{0}$, so that $\mathbf{x}$ is the unique minimizer.

%On the other hand, if $\mathbf{x}$ is the unique minimizer of (\ref{eqProblemP1}), then $\norm{\mathbf{\Omega}(\mathbf{x}+\mathbf{v})}_1>\norm{\mathbf{\Omega x}}_1$ for all $\mathbf{v}\in\ker\mathbf{M}\setminus\{\mathbf{0}\}$, which implies that $\mathbf{v}\notin T(\mathbf{x})$. This means that 
%\[
%\brac{\ker\mathbf{M}\setminus\{\mathbf{0}\}}\cap T(\mathbf{x})=\eset
%\]
%or equivalently $\ker\mathbf{M}\cap T(\mathbf{x})=\{\mathbf{0}\}$. 
The other direction is proved by contradiction. Let $\mathbf x$ be the unique minimizer of (\ref{eqProblemP1}). Take any $\mathbf v\in T(\mathbf x)\setminus\{\mathbf 0\}$. Then $\mathbf v$ can be written as
\[
\mathbf v=\sum_{j}t_j(\mathbf{z}_j-\mathbf x),\quad t_j\geq 0,\quad \norm{\mathbf{\Omega}\mathbf z_j}_1\leq\norm{\mathbf{\Omega x}}_1.
\]
Since $\mathbf v\neq \mathbf 0$, it holds $\sum\limits_{j}t_j>0$ and we can define $t_j':=\frac{t_j}{\sum\limits_{j}t_j}$. Suppose $\mathbf v\in\ker\mathbf M$. Then
\[
\mathbf 0=\mathbf M\brac{\frac{\mathbf v}{\sum\limits_{j}t_j}}=\mathbf M\brac{\sum_jt_j'\mathbf z_j}-\mathbf{Mx},
\]
so that $\mathbf M\brac{\sum_jt_j'\mathbf z_j}=\mathbf{Mx}$. Together with the estimate 
\[
\norm{\sum\limits_jt_j'\mathbf{z}_j}_1\leq\sum_jt_j'\norm{\mathbf z_j}\leq\norm{\mathbf x}_1
\]
and uniqueness of the minimizer, this implies $\sum\limits_jt_j'\mathbf{z}_j=\mathbf x$. Hence $\mathbf v=\mathbf 0$, which leads to a contradiction. Thus, $\ker\mathbf M\cap T(\mathbf x)=\{\mathbf 0\}$. 
\end{proof}
When the measurements are noisy, we use the following condition for successful recovery \cite{ChandrasekaranRechtParriloWillsky}.
\begin{theorem}\label{thRecoveryViaTangentConesWithNoise}
Let $\mathbf{x}\in\RR^d$, $\mathbf{M}\in\RR^{m\times d}$ and $\mathbf y = \mathbf{M}\mathbf{x}+\mathbf{w}$ with $\norm{\mathbf{w}}_2\leq\eta$. If
\begin{equation}\label{eqInfOverCone}
\underset{\begin{subarray}{c}
\mathbf{v}\in T(\mathbf{x})\\
\norm{\mathbf{v}}_2=1
 \end{subarray}}\inf\norm{\mathbf{Mv}}_2\geq\tau
\end{equation}
for some $\tau$, then a minimizer $\mathbf{\hat{x}}$ of (\ref{eqProblemP1Noise}) satisfies
\[
\norm{\mathbf{x}-\mathbf{\hat{x}}}_2\leq\frac{2\eta}{\tau}.
\]
\end{theorem}
\begin{proof}
Since $\mathbf{\hat x}$ is a minimizer of (\ref{eqProblemP1Noise}), we have $\norm{\mathbf{\mathbf{\Omega\hat x}}}_1\leq\norm{\mathbf{\Omega x}}_1$ and $\mathbf{\hat x}-\mathbf{x}\in T(\mathbf{x})$. Our assumption (\ref{eqInfOverCone}) implies
\begin{equation}\label{eqMeasurementOfDifferenceInCone}
\norm{\mathbf{M}(\mathbf{\hat{x}}-\mathbf{x})}_2\geq\tau\norm{\mathbf{\hat{x}}-\mathbf{x}}_2.
\end{equation}
On the other hand, we can give an upper bound for $\norm{\mathbf{M\hat{x}}-\mathbf{M}\mathbf{x}}_2$ by
\begin{equation}\label{eqMeasurementOfDIfferenceInitialCondition}
\norm{\mathbf{M\hat x}-\mathbf{M}\mathbf{x}}_2\leq\norm{\mathbf{M\hat x}-\mathbf{y}}+\norm{\mathbf{M}\mathbf{x}-\mathbf{y}}_2\leq 2\eta.
\end{equation}
Combining (\ref{eqMeasurementOfDifferenceInCone}) and (\ref{eqMeasurementOfDIfferenceInitialCondition}) we get the desired estimate. 
\end{proof}

%%%%%%%%%%%%%%%%%%%%%%%%%%%%%%%%%%%%%%%%%%%%%%%%%%%%%%%%

\subsection{Nonuniform recovery with Gaussian measurements}
\label{secNonUniformRecovery}

To prove the non-uniform recovery result for Gaussian random measurements (Theorem \ref{thMainResultForFrame}) we rely on the condition stated in Theorem \ref{thRecoveryViaTangentCones}, which requires that the null space of the measurement matrix $\mathbf{M}$ misses the set $T(\mathbf{x})$. The next ingredient of the proof is a variation 
of the Gordon's escape through a mesh theorem, which was first used in the context of compressed sensing in \cite{RudelsonVershynin}.
To state this theorem, we introduce some notation and formulate auxiliary lemmas.

Let $\mathbf{g}\in\RR^m$ be a standard Gaussian random vector, that is, a vector of independent mean zero, variance one normal distributed random variables. 
Then for 
\[
E_m:=\mean\norm{\mathbf{g}}_2=\sqrt{2}\;\frac{\Gamma\brac{(m+1)/2}}{\Gamma\brac{m/2}}
\]  
we have
\[
\frac{m}{\sqrt{m+1}}\leq E_m\leq\sqrt{m},
\]
see \cite{Gordon,FoucartRauhut}. For a set $T\subset\RR^d$ we define its Gaussian width by 
\[
\ell(T):=\mean\underset{\mathbf{x}\in T}\sup\abrac{\mathbf{x},\mathbf{g}},
\]
where $\mathbf{g}\in\RR^d$ is a standard Gaussian random vector.
\begin{lemma}[Gordon \cite{Gordon}]\label{lmGordon}
Let $X_{i,j}$ and $Y_{i,j}$, $i=1\ldots,m$, $j=1,\ldots,n$ be two mean-zero Gaussian random variables. If
\[
\begin{aligned}
\mean\abs{X_{i,j}-X_{k,l}}^2&\leq\mean\abs{Y_{i,j}-Y_{k,l}}^2,\;\;\mbox{for all}\;\;i\neq k,\;\mbox{and}\;j,l\\
\mean\abs{X_{i,j}-X_{i,l}}^2&\geq\mean\abs{Y_{i,j}-Y_{i,l}}^2\;\;\mbox{for}\;i,j,l,
\end{aligned}
\]
then
\[
\mean\underset{i}\min\,\underset{j}\max\,X_{i,j}\geq\mean\underset{i}\min\,\underset{j}\max\,Y_{i,j}.
\]
\end{lemma}
\begin{remark}\label{Gordon:inf}
Gordon's lemma extends to the case of Gaussian processes indexed by possibly infinite sets where the expected maxima or minima are replaced
by corresponding lattice suprema or infima, see for instance \cite[Remark 8.28]{FoucartRauhut} or \cite{leta91}.
\end{remark}

We further exploit the concentration of measure phenomenon, which asserts that Lipschitz functions concentrate well around their expectation \cite{Ledoux,Massart}.
\begin{lemma}[Concentration of measure]\label{lmConcentrationOfMeasure}
Let $f:\RR^n\to\RR$ be an $L$-Lipschitz function:
\[
\abs{f(\mathbf{x})-f(\mathbf{y})}\leq L\norm{\mathbf{x}-\mathbf{y}}_2,\;\;\mbox{for all}\;\mathbf{x},\mathbf{y}\in\RR^n.
\]
Let $\mathbf{g}=(g_1,g_2,\ldots,g_n)$ be a vector of independent standard normal random variables. Then, for all $t>0$,
\[
\PP\brac{\mean[f(\mathbf{g})]-f(\mathbf{g})\geq t}\leq e^{-\frac{t^2}{2L^2}}.
\]
\end{lemma}

Next, we state our modification of the Gordon's escape through a mesh theorem, see \cite{Gordon} for the original version. Below, 
$\mathbf{\Omega}(T)$ corresponds to the set of elements produced by applying $\mathbf{\Omega}$ to the elements of $T$.

\begin{theorem}\label{thModifiedGordonsEscapeThroughTheMesh}
Let $\mathbf{\Omega}\in\RR^{p\times d}$ be a frame with constants $A$, $B>0$. Let $\mathbf{M}\in\RR^{m\times d}$ be a Gaussian random matrix and $T$ be a subset of the unit sphere $\mathbb{S}^{d-1}=\{\mathbf{x}\in\RR^d: \norm{\mathbf{x}}_2=1\}$. Then, for $t>0$, it holds
\begin{equation}\label{eqGordonsEscapeThroughTheMesh}
\PP\brac{\underset{\mathbf{x}\in T}\inf\norm{\mathbf{M}\mathbf{x}}_2> E_m-\frac{1}{\sqrt{A}}\ell\brac{\mathbf{\Omega}(T)}-t}\geq 1-e^{-\frac{t^2}{2}}.
\end{equation}
\end{theorem}
\begin{proof}
Recall that 
\[
\norm{\mathbf{M}\mathbf{x}}_2=\underset{\mathbf{y}\in S^{m-1}}\max\abrac{\mathbf{M}\mathbf{x},\mathbf{y}}.
\]
For $\mathbf{x}\in T$ and $\mathbf{y}\in S^{m-1}$ we compare the two Gaussian processes 
\[
X_{\mathbf{x},\mathbf{y}}:=\abrac{\mathbf{M}\mathbf{x},\mathbf{y}}\qquad \mbox{and}\qquad Y_{\mathbf{x},\mathbf{y}}:=\frac{1}{\sqrt{A}}\abrac{\mathbf{g},\mathbf{\Omega x}}+\abrac{\mathbf{h},\mathbf{y}},
\]
where $\mathbf{g}\in\RR^{p}$ and $\mathbf{h}\in\RR^m$ are independent standard Gaussian random vectors. Let $\mathbf{x},\mathbf{x'}\in \mathbb{S}^{d-1}$ and $\mathbf{y},\mathbf{y'}\in S^{m-1}$. Since $M_{ij}$ are independent with $\mean M_{ij} = 0$, $\mean M_{ij}^2=1$, we have  
\begin{align}
&\mean\abs{X_{\mathbf{x},\mathbf{y}}-X_{\mathbf{x'},\mathbf{y'}}}^2=\mean\abs{\sum_{i=1}^m\sum_{j=1}^dM_{ij}(x_jy_i-x'_jy'_i)}^2=\sum_{i=1}^m\sum_{j=1}^d(x_jy_i-x'_jy'_i)^2\notag\\
&=\sum_{i=1}^m\sum_{j=1}^d(x_j^2y_i^2+x'^2_jy'^2_i-2x_jx'_jy_iy'_i)=\norm{\mathbf{x}}_2^2\norm{\mathbf{y}}_2^2+\norm{\mathbf{x'}}_2^2\norm{\mathbf{y'}}_2^2-2\abrac{\mathbf{x},\mathbf{x'}}\abrac{\mathbf{y},\mathbf{y'}}\notag\\
&=2-2\abrac{\mathbf{x},\mathbf{x'}}\abrac{\mathbf{y},\mathbf{y'}}\label{eqEstimateForXprocess}.
\end{align}
Independence and the isotropicity of the Gaussian vectors $\mathbf{g}$ and $\mathbf{h}$ together with the fact that $\mathbf{\Omega}$ is a frame with lower frame bound $A$ imply
\begin{align}
\mean\abs{Y_{\mathbf{x},\mathbf{y}}-Y_{\mathbf{x'},\mathbf{y'}}}^2  &=\mean\abs{\frac{1}{\sqrt{A}}\abrac{\mathbf{g},\mathbf{\Omega x} -\mathbf{\Omega x}'}}^2+\mean\abs{\abrac{\mathbf{h},\mathbf{y}-\mathbf{y}'}}^2\notag\\
&=\frac{1}{A}\norm{\mathbf{\Omega x} -\mathbf{\Omega x'}}_2^2+\norm{\mathbf{y}-\mathbf{y'}}_2^2\geq \norm{\mathbf{x}-\mathbf{x'}}_2^2+\norm{\mathbf{y}-\mathbf{y}'}_2^2\notag\\
&=\norm{\mathbf{x}}_2^2+\norm{\mathbf{x'}}_2^2-2\abrac{\mathbf{x},\mathbf{x'}}+\norm{\mathbf{y}}_2^2+\norm{\mathbf{y'}}_2^2-2\abrac{\mathbf{y},\mathbf{y'}}\notag\\
\label{eqEstimateForYProcess}
&=4-2\abrac{\mathbf{x},\mathbf{x'}}-2\abrac{\mathbf{y},\mathbf{y'}}.
\end{align}
When $\mathbf{x}=\mathbf{x'}$, we have
\[
\mean\abs{Y_{\mathbf{x},\mathbf{y}}-Y_{\mathbf{x},\mathbf{y'}}}^2=\norm{\mathbf{y}-\mathbf{y}'}_2^2=2-2\abrac{\mathbf{y},\mathbf{y'}}.
\]
Combining (\ref{eqEstimateForXprocess}) and (\ref{eqEstimateForYProcess}), we obtain
\[
\mean\abs{Y_{\mathbf{x},\mathbf{y}}-Y_{\mathbf{x'},\mathbf{y'}}}^2-\mean\abs{X_{\mathbf{x},\mathbf{y}}-X_{\mathbf{x'},\mathbf{y'}}}^2\geq 2(1-\abrac{\mathbf{x},\mathbf{x'}})(1-\abrac{\mathbf{y},\mathbf{y'}})
\]
and since
$\abrac{\mathbf{x},\mathbf{x'}}\leq\norm{\mathbf{x}}_2\norm{\mathbf{x'}}_2\leq 1$ and similarly for $\mathbf{y},\mathbf{y}'$, it follows that
\[
\mean\abs{Y_{\mathbf{x},\mathbf{y}}-Y_{\mathbf{x'},\mathbf{y'}}}^2-\mean\abs{X_{\mathbf{x},\mathbf{y}}-X_{\mathbf{x'},\mathbf{y'}}}^2\geq 0.
\]
Moreover, we have
\[
\mean\abs{Y_{\mathbf{x},\mathbf{y}}-Y_{\mathbf{x},\mathbf{y'}}}^2=\mean\abs{X_{\mathbf{x},\mathbf{y}}-X_{\mathbf{x},\mathbf{y'}}}^2.
\]
Due to Gordon's lemma (Lemma \ref{lmGordon}) and Remark \ref{Gordon:inf}, 
\begin{align}
\mean&\underset{\mathbf{x}\in T}\inf\norm{\mathbf{M}\mathbf{x}}_2=\mean\underset{\mathbf{x}\in T}\inf\underset{\mathbf{y}\in S^{m-1}}\max X_{\mathbf{x},\mathbf{y}}\geq\mean\underset{\mathbf{x}\in T}\inf\underset{\mathbf{y}\in S^{m-1}}\max Y_{\mathbf{x},\mathbf{y}}\notag\\
&=\mean\underset{\mathbf{x}\in T}\inf\underset{\mathbf{y}\in S^{m-1}}\max\fbrac{\frac{1}{\sqrt{A}}\abrac{\mathbf{g},\mathbf{\Omega x}}+\abrac{\mathbf{h},\mathbf{y}}}=\mean\underset{\mathbf{x}\in T}\inf\fbrac{\frac{1}{\sqrt{A}}\abrac{\mathbf{g},\mathbf{\Omega x}}+\norm{\mathbf{h}}_2}\notag\\
&=\mean\norm{\mathbf{h}}_2-\frac{1}{\sqrt{A}}\mean\underset{\mathbf{x}\in T}\sup\abrac{\mathbf{g},\mathbf{\Omega x}}=E_m-\frac{1}{\sqrt{A}}\mean\underset{\mathbf{z}\in\mathbf{\Omega}(T)}\sup\abrac{\mathbf{g},\mathbf{z}}=E_m-\frac{1}{\sqrt{A}}\ell(\mathbf{\Omega}(T))\label{eqExpectationOfLipschitzFunctionEstimate}.
\end{align}
Let $F(\mathbf{M}):=\underset{\mathbf{x}\in T}\inf\norm{\mathbf{M}\mathbf{x}}_2$. For any $\mathbf{A},\mathbf{B}\in\RR^{m\times d}$
\begin{equation*}
\underset{\mathbf{x}\in T}\inf\norm{\mathbf{A}\mathbf{x}}_2\leq\underset{\mathbf{x}\in T}\inf\brac{\norm{\mathbf{B}\mathbf{x}}_2+\norm{\brac{\mathbf{A}-\mathbf{B}}\mathbf{x}}_2}
\]
\[
\leq\underset{\mathbf{x}\in T}\inf\norm{\mathbf{B}\mathbf{x}}_2+\norm{\mathbf{A}-\mathbf{B}}_{2\to2} \leq \underset{\mathbf{x}\in T}\inf\norm{\mathbf{B}\mathbf{x}}_2+\norm{\mathbf{A}-\mathbf{B}}_{F}.
\end{equation*}
The second inequality follows from the fact that $T\subset \mathbb{S}^{d-1}$. By interchanging $\mathbf{A}$ and $\mathbf{B}$ we conclude that 
\[
\abs{F(\mathbf{A})-F(\mathbf{B})}\leq\norm{\mathbf{A}-\mathbf{B}}_{F}.
\]
This means that $F$ is $1$-Lipschitz with respect to the Frobenius norm (which corresponds to the $\ell_2$-norm when interpreting a matrix as a vector) 
and due to concentration of measure (Lemma~\ref{lmConcentrationOfMeasure})
\[
\PP\brac{\underset{\mathbf{x}\in T}\inf\norm{\mathbf{M}\mathbf{x}}_2\leq\mean\underset{\mathbf{x}\in T}\inf\norm{\mathbf{M}\mathbf{x}}_2-t}\leq e^{-t^2/2}.
\]
Applying the estimate (\ref{eqExpectationOfLipschitzFunctionEstimate}) to the previous inequality gives
\[
\PP\!\brac{\underset{\mathbf{x}\in T}\inf\norm{\mathbf{M}\mathbf{x}}_2\leq E_m\!-\!\frac{\ell(\mathbf{\Omega}(T))}{\sqrt{A}}-t}\leq\PP\!\brac{\underset{\mathbf{x}\in T}\inf\norm{\mathbf{M}\mathbf{x}}_2\leq\mean\!\underset{\mathbf{x}\in T}\inf\norm{\mathbf{M}\mathbf{x}}_2-t}\leq e^{-\frac{t^2}{2}},
\]
which concludes the proof. 
\end{proof}
The previous result suggests to estimate the Gaussian width of $\mathbf{\Omega}(T)$ with $T := T(\mathbf{x})\cap \mathbb{S}^{d-1}$.
Since $\mathbf{\Omega}$ is a frame with upper frame constant $B$, we have
\[
\mathbf{\Omega}(T)\subset\mathbf{\Omega}(T(\mathbf{x}))\cap\mathbf{\Omega}(\mathbb{S}^{d-1})\subset K(\mathbf{\Omega x})\cap\brac{\sqrt{B}\mathbb{B}_2^p},
\]
where
\[
K(\mathbf{\Omega x})=\cone\fbrac{\mathbf{y}-\mathbf{\Omega x}:\mathbf{y}\in\RR^p,\; \norm{\mathbf{y}}_1\leq\norm{\mathbf{\Omega x}}_1}.
\]
The supremum over a larger set can only increase, hence
\begin{equation}\label{eqGWidthByKoneAndSphere}
\ell(\mathbf{\Omega}(T))\leq\sqrt B\ell\brac{K(\mathbf{\Omega x})\cap \mathbb{B}_2^{p}}.
\end{equation}
We next recall an upper bound for the Gaussian width $\ell\brac{K(\mathbf{\Omega x})\cap \mathbb{B}_2^p}$ from  \cite{ChandrasekaranRechtParriloWillsky} 
involving the polar cone $\mathcal{N}(\mathbf{\Omega x})=K(\mathbf{\Omega x})^{\circ}$ defined by
\[
\mathcal{N}(\mathbf{\Omega x})=\left\{\mathbf{z}\in\RR^{p}:\abrac{\mathbf{z},\mathbf{y}-\mathbf{\Omega x}}\leq 0\;\mbox{for all}\;\mathbf{y}\in\RR^p\;\;\mbox{such that}\;\norm{\mathbf{y}}_1\leq\norm{\mathbf{\Omega x}}_1\right\}.
\]
\begin{proposition}\label{prGaussianWidthsByPolarCone}
Let $\mathbf{g}\in\RR^{p}$ be a standard Gaussian random vector. Then
\[
\ell\brac{K(\mathbf{\Omega x})\cap \mathbb{B}_2^p}\leq\mean\underset{\mathbf{z}\in\mathcal{N}(\mathbf{\Omega x})}\min\norm{\mathbf{g}-\mathbf{z}}_2.
\]
\end{proposition}
The proof relies on tools from convex analysis, see \cite{ChandrasekaranRechtParriloWillsky,FoucartRauhut}, \cite[Ch.\ 5.9]{BoydVandenberghe}.  
\begin{proposition}
Let $s$ be the sparsity of the vector $\mathbf{\Omega x}\in\RR^p$. Then
\begin{equation}\label{eqGaussianWidthOfCone}
\ell\brac{K(\mathbf{\Omega x})\cap \mathbb{B}_2^p}^2\leq 2s\ln\frac{ep}{s}.
\end{equation}
\end{proposition}
\begin{proof}
By Proposition \ref{prGaussianWidthsByPolarCone} and H\"older's inequality
\begin{equation}
\ell\brac{K(\mathbf{\Omega x})\cap \mathbb{B}_2^p}^2\leq\brac{\mean\underset{\mathbf{z}\in\mathcal{N}(\mathbf{\Omega x})}\min\norm{\mathbf{g}-\mathbf{z}}_2}^2\leq\mean\underset{\mathbf{z}\in \mathcal{N}(\mathbf{\Omega x})}\min\norm{\mathbf{g}-\mathbf{z}}_2^2.
\end{equation}
Let $S$ denote the support of $\mathbf{\Omega x}$. Then one can verify that 
\begin{equation}\label{eqPolarConeAsUnionOverT}
\mathcal{N}(\mathbf{\Omega x}) = \bigcup_{t\geq 0} \fbrac{\mathbf{z}\in\RR^p:\,z_i=t\sgn(\mathbf{\Omega x})_i,\;i\in S,\;\abs{z_i}\leq t,\;i\in S^c},
\end{equation}
see \cite[Lemma 9.23]{FoucartRauhut} for a proof. To proceed, we fix $t$, minimize $\norm{\mathbf{g}-\mathbf{z}}_2^2$ over all possible entries $z_j$, take the expectation of the obtained expression and finally optimize over $t$.
According to (\ref{eqPolarConeAsUnionOverT}), we have 
\[
\begin{aligned}
\underset{\mathbf{z}\in\mathcal{N}(\mathbf{\Omega x})}\min\norm{\mathbf{g}-\mathbf{z}}_2^2 &= \underset{\begin{subarray}{c}
t\geq 0\\
\abs{z_i}\leq t,\,i\in S^c
\end{subarray}}\min\sum_{i\in S}\brac{g_i-t\sgn(\mathbf{\Omega x})_i}^2+\sum_{i\in S^c}\brac{g_i-z_i}^2\\
&=\underset{\begin{subarray}{c}
t\geq 0
\end{subarray}}\min\sum_{i\in S}\brac{g_i-t\sgn(\mathbf{\Omega x})_i}^2+\sum_{i\in S^c}S_t(g_i)^2,
\end{aligned}
\]
where $S_t$ is the soft-thresholding operator given by
\[
S_t(x)=\left\{\begin{array}{ll}
x+t, & x<-t,\\
0, & -t\leq x\leq t,\\
x-t, & x>t.
\end{array}\right.
\]
Taking expectation we arrive at
\begin{align}
\mean\underset{\mathbf{z}\in\mathcal{N}(\mathbf{\Omega x})}\min\norm{\mathbf{g}-\mathbf{z}}_2^2 &\leq \mean\sbrac{\sum_{i\in S}\brac{g_i-t\sgn(\mathbf{\Omega x})_i}^2}+\mean\sbrac{\sum_{i\in S^c}S_t(g_i)^2}\notag\\
&=s(1+t^2)+(p-s)\mean S_t(g)^2\label{eqExpectationEstimateWithT}, 
\end{align}
where $g$ is a univariate standard Gaussian random variable. To calculate the expectation of $S_t(g)^2$, we apply the direct integration
\begin{align}
\mean S_t(g)^2 & =\frac{1}{\sqrt{2\pi}}\sbrac{\int\limits_{-\infty}^{-t}(x+t)^2e^{-\frac{x^2}{2}}\,dx+\int\limits_t^{\infty}(x-t)^2e^{-\frac{x^2}{2}}\,dx}\notag\\
& = \frac{2}{\sqrt{2\pi}}\int\limits_{0}^{\infty} x^2e^{-\frac{(x+t)^2}{2}}\,dx
=\frac{2e^{-\frac{t^2}{2}}}{\sqrt{2\pi}}\int\limits_{0}^{\infty} x^2e^{-\frac{x^2}{2}}e^{-xt}\,dx\notag\\
&\leq e^{-\frac{t^2}{2}}\sqrt{\frac{2}{\pi}}\int\limits_{0}^{\infty} x^2e^{-\frac{x^2}{2}}\,dx=e^{-\frac{t^2}{2}}.\label{eqExpectationOfSoftThreshold}
\end{align}
Substituting the estimate (\ref{eqExpectationOfSoftThreshold}) into (\ref{eqExpectationEstimateWithT}) gives
\[
\mean\underset{\mathbf{z}\in\mathcal{N}(\mathbf{\Omega x})}\min\norm{\mathbf{g}-\mathbf{z}}_2^2\leq s(1+t^2)+(p-s)e^{-\frac{t^2}{2}}. 
\]
Setting $t=\sqrt{2\ln(p/s)}$ finally leads to
\[
\ell\brac{K(\mathbf{\Omega x})\cap \mathbb{B}_2^p}^2\leq s\brac{1+2\ln(p/s)}+s=2s\ln(ep/s).
\]
This concludes the proof. 
\end{proof}

By combining inequalities (\ref{eqGWidthByKoneAndSphere}) and (\ref{eqGaussianWidthOfCone}) we obtain
\[
\ell(\mathbf{\Omega}(T))^2\leq 2Bs\ln\frac{ep}{s}.
\]

\begin{proof}[of Theorem \ref{thMainResultForFrame}]
Set $t=\sqrt{2\ln(\eps^{-1})}$. The fact that $E_m\geq m/\sqrt{m+1}$ along with condition (\ref{eqNumberOfMeasurementsForFrame}) yields 
\[
E_m\geq \frac{1}{\sqrt{A}}\ell(\mathbf{\Omega}(T))+t.
\]
Theorem~\ref{thModifiedGordonsEscapeThroughTheMesh} implies
\[
\PP\brac{\underset{\mathbf{x}\in T}\inf\norm{\mathbf{M}\mathbf{x}}_2> 0}\geq\PP\brac{\underset{\mathbf{x}\in T}\inf\norm{\mathbf{M}\mathbf{x}}_2> E_m-\frac{1}{\sqrt{A}}\ell\brac{\mathbf{\Omega}(T)}-t}\geq\ 1-e^{-\frac{t^2}{2}}=1-\eps,
\]
which guarantees that $\ker\mathbf{M}\cap T(\mathbf{x})=\{\mathbf 0\}$ with probability at least $1-\eps$. As the final step we apply Theorem \ref{thRecoveryViaTangentCones}. 
\end{proof}

We now extend Theorem~\ref{thMainResultForFrame} to robust recovery.
\begin{theorem}\label{thNoisyMeasurements} Let $\mathbf{\Omega} \in \mathbb{R}^{p \times d}$ be a frame with frame bounds $A,B > 0$ and
let $\mathbf{x}$ be $l$-cosparse and $s = p -l$. For a random draw $\mathbf{M}\in\RR^{m\times d}$ of a Gaussian random matrix, 
let noisy measurements $\mathbf{y}=\mathbf{M}\mathbf{x}+\mathbf{w}$ be given with $\norm{\mathbf{w}}_2\leq\eta$. If for $0<\eps<1$ and some $\tau>0$
\begin{equation}\label{eqNumberOfMeasurementsForFrameNoise}
\frac{m^2}{m+1}\geq \frac{2Bs}{A}\brac{\sqrt{\ln\frac{ep}{s}}+\sqrt{\frac{A\ln(\eps^{-1})}{Bs}}+\tau\sqrt{\frac{A}{2sB}}}^2,
\end{equation}
then with probability at least $1-\eps$, any minimizer $\mathbf{\hat x}$ of (\ref{eqProblemP1Noise}) satisfies
\[
\norm{\mathbf{x}-\mathbf{\hat{x}}}_2\leq\frac{2\eta}{\tau}.
\] 
\end{theorem}
\begin{proof}
We use the recovery condition stated in Theorem \ref{thRecoveryViaTangentConesWithNoise}. Set $t=\sqrt{2\ln(\eps^{-1})}$. Our previous considerations  and the choice of $m$ in (\ref{eqNumberOfMeasurementsForFrameNoise}) guarantee that
\[
E_m-\frac{1}{\sqrt{A}}\ell(\mathbf{\Omega}(T))-t\geq\frac{m}{\sqrt{m+1}}-\sqrt{\frac{2Bs}{A}\ln\frac{ep}{s}}-\sqrt{2\ln(\eps^{-1})}\geq\tau.
\]  
The monotonicity of probability and Theorem \ref{thModifiedGordonsEscapeThroughTheMesh} yield
\[
\PP\brac{\underset{\mathbf{x}\in T}\inf\norm{\mathbf{M}\mathbf{x}}_2\geq\tau}\geq\PP\brac{\underset{\mathbf{x}\in T}\inf\norm{\mathbf{M}\mathbf{x}}_2\geq E_m-\frac{1}{\sqrt{A}}\ell(\mathbf{\Omega}(T))-t}\geq 1-\eps.
\]
\end{proof}

%%%%%%%%%%%%%%%%%%%%%%%%%%%%%%%%%%%%%%%%%%%%%%%%%%%
%%%%%%%%%%%%%%%%%%%%%%%%%%%%%%%%%%%%%%%%%%%%%%%%%%%

\section{Uniform recovery}

This section is dedicated to the proof of the uniform recovery result in Theorem~\ref{thUniformRecoveryWithFrame}. 
It relies on the $\mathbf{\Omega}$-null space property, which extends the null space property known
from the standard synthesis sparsity case, see e.g.~\cite{codade09,FoucartRauhut,grni03}. 
We analyze this property directly for Gaussian random matrices with similar techniques as used in the previous section.

\subsection{$\mathbf{\Omega}$-null space property}

Let us start with the $\mathbf{\Omega}$-null space property which is a sufficient condition for the exact reconstruction of every cosparse vector. 
\begin{definition}\label{defNSP}
A matrix $\mathbf{M}\in\RR^{m\times d}$ is said to satisfy the $\mathbf{\Omega}$-null space property of order $s$ with constant $0<\rho<1$, if for any set $\Lambda\subset [p]$ with $\# \Lambda\geq p-s$ it holds
\begin{equation}\label{eqNSP}
\norm{\mathbf{\Omega}_{\Lambda^c}\mathbf{v}}_1\leq\rho\norm{\mathbf{\Omega}_{\Lambda}\mathbf{v}}_1\;\;\;\mbox{for all}\;\;\mathbf{v}\in\ker{\mathbf M}.
\end{equation}
\end{definition}
If $\mathbf{\Omega}$ is the identity map $\Id:\RR^d\to\RR^d$, then (\ref{eqNSP}) is the standard null space property. We start with a result on exact recovery of cosparse vectors.
\begin{theorem}\label{thRecoveryWithOmegaNSP}
If $\mathbf{M}\in\RR^{m\times d}$ satisfies the $\mathbf{\Omega}$-null space property of order $s$ with $0<\rho<1$, then every $l$-cosparse vector $\mathbf{x}\in\RR^d$ with $l=p-s$ is the unique solution of (\ref{eqProblemP1}) with $\mathbf{y}=\mathbf{Mx}$.
\end{theorem}
This theorem follows immediately from the next result, which also implies a certain stability estimate in $\ell_1$.

 \begin{theorem}\label{thConeCostraint}
Let $\mathbf{x}\in\RR^d$ be an arbitrary vector and $\mathbf{\hat x}$ be a solution of (\ref{eqProblemP1}) with $\mathbf{y}=\mathbf{M}\mathbf{x}$, where $\mathbf{M}\in\RR^{m\times d}$ 
satisfies the $\mathbf{\Omega}$-null space property of order $s$ with constant $\rho \in (0,1)$. Then
\begin{equation}\label{eqConeConstraint}
\norm{\mathbf{\Omega}\brac{\mathbf{x}-\mathbf{\hat{x}}}}_1\leq\frac{2(1+\rho)}{1-\rho} \sigma_s(\mathbf{\Omega x})_1.
\end{equation}
\end{theorem}
\begin{proof}
Since $\mathbf{\hat x}$ is the solution of (\ref{eqProblemP1}), we must have $\norm{\mathbf{\Omega\hat x}}_1\leq\norm{\mathbf{\Omega x}}_1$. Take any  $\Lambda\subset [p]$ with $\# \Lambda\geq p-s$. Then
\[
\norm{\mathbf{\Omega}_{\Lambda^c}\mathbf{\hat x}}_1+\norm{\mathbf{\Omega}_{\Lambda}\mathbf{\hat x}}_1\leq\norm{\mathbf{\Omega}_{\Lambda^c}\mathbf{x}}_1+\norm{\mathbf{\Omega}_{\Lambda}\mathbf{x}}_1.
\]
By the triangle inequality, the vector $\mathbf v := \mathbf{x}-\mathbf{\hat{x}}$ satisfies
\[
\norm{\mathbf{\Omega}_{\Lambda^c}\mathbf{x}}_1-\norm{\mathbf{\Omega}_{\Lambda^c} \mathbf{v}}_1+\norm{\mathbf{\Omega}_{\Lambda} \mathbf{v}}_1-\norm{\mathbf{\Omega}_{\Lambda}\mathbf{x}}_1\leq\norm{\mathbf{\Omega}_{\Lambda^c}\mathbf{x}}_1+\norm{\mathbf{\Omega}_{\Lambda}\mathbf{x}}_1,
\]
which implies
\[
\norm{\mathbf{\Omega}_{\Lambda} \mathbf{v}}_1\leq\norm{\mathbf{\Omega}_{\Lambda^c} \mathbf{v}}_1+2\norm{\mathbf{\Omega}_{\Lambda}\mathbf{x}}_1 
\leq \rho \norm{\mathbf{\Omega}_{\Lambda}\mathbf{v}}_1 + 2\norm{\mathbf{\Omega}_{\Lambda}\mathbf{x}}_1. 
\]
Hereby, we have applied the $\mathbf{\Omega}$-null space property \eqref{eqNSP}. Rearranging and choosing a set $\Lambda$ of size $p-s$ which minimizes $
\norm{\mathbf{\Omega}_{\Lambda}\mathbf{x}}_1$ yields 
\[
\norm{\mathbf{\Omega}_{\Lambda}\mathbf{v}}_1\leq\frac{2}{1-\rho} \sigma_s(\mathbf{\Omega x})_1.
\]
Furthermore, another application of the $\mathbf{\Omega}$-null space property gives
\[
\norm{\mathbf{\Omega}\mathbf{v}}_1 = \norm{\mathbf{\Omega}_{\Lambda}\mathbf{v}}_1 +  \norm{\mathbf{\Omega}_{\Lambda^c}\mathbf{v}}_1
\leq (1+\rho)  \norm{\mathbf{\Omega}_{\Lambda}\mathbf{v}}_1  \leq \frac{2(1+\rho)}{1-\rho} \sigma_s(\mathbf{\Omega x})_1.
\]
This completes the proof. 
\end{proof}

In order to provide a suitable stability estimate in $\ell_2$ we require a slightly stronger version
of the $\mathbf\Omega$-null space property. 
\begin{definition}\label{defL2StableNSP}
A matrix $\mathbf{M}\in\RR^{m\times d}$ is said to satisfy the $\ell_2$-stable $\mathbf{\Omega}$-null space property of order $s$ with constant $0<\rho<1$, if, for any set $\Lambda\subset [p]$ with $\# \Lambda\geq p-s$, it holds
\begin{equation}\label{eqL2NSP}
\norm{\mathbf{\Omega}_{\Lambda^c}\mathbf{v}}_2\leq\frac{\rho}{\sqrt{s}}\norm{\mathbf{\Omega}_{\Lambda}\mathbf{v}}_1\;\;\;\mbox{for all}\;\;\mathbf{v}\in\ker{\mathbf M}.
\end{equation}
\end{definition}
\begin{remark}\label{rem:l1l2}
The H\"older's inequality implies $\norm{\mathbf\Omega_{\Lambda^c}\mathbf v}_1\leq \sqrt s\norm{\mathbf\Omega_{\Lambda^c}\mathbf v}_2$ for any set $\Lambda \subset [p]$ with $\#(\Lambda^c) = s$.
This means that if $\mathbf M\in\RR^{m\times d}$ satisfies the $\ell_2$-stable $\mathbf\Omega$-null space property of order $s$ with constant $0<\rho<1$, then it satisfies the $\mathbf\Omega$-null space property of the same order and with the same constant.
\end{remark}
\begin{theorem}\label{thRecoveryWithL2NSP}
Let $\mathbf{M}\in\RR^{m\times d}$ satisfy the $\ell_2$-stable $\mathbf{\Omega}$-null space property of order $s$ with constant $0<\rho<1$. Then for any $\mathbf{x}\in\RR^d$ the solution $\mathbf{\hat x}$ of (\ref{eqProblemP1}) with $\mathbf{y}=\mathbf{M}\mathbf{x}$ approximates the vector $\mathbf{x}$ with $\ell_2$-error
\begin{equation}\label{eqL2StableRecovery}
\norm{\mathbf{x}-\mathbf{\hat{x}}}_2\leq\frac{2(1+\rho)^2}{\sqrt{A}(1-\rho)}\frac{\sigma_{s}(\mathbf{\Omega x})_1}{\sqrt{s}}.
\end{equation}
\end{theorem}
Inequality (\ref{eqL2StableRecovery}) means that $l$-cosparse vectors are exactly recovered by (\ref{eqProblemP1}) and vectors $\mathbf{x}\in\RR^d$, such that $\mathbf{\Omega x}$ is close to an $s$-sparse vector in $\ell_1$, can be well approximated in $\ell_2$ by a solution of (\ref{eqProblemP1}). The proof goes along the same lines as in the standard case in \cite{FoucartRauhut}. The novelty here is that we exploit the sparsity not of the signal itself, but of its analysis representation. So first we extend the $\ell_1$-error estimate above to an $\ell_2$-error estimate for $\mathbf{\Omega x}$ 
and use the fact that $\mathbf\Omega$ is a frame to bound the $\ell_2$-error $\norm{\mathbf x-\mathbf{\hat x}}_2$. The statement of Theorem \ref{thRecoveryWithL2NSP} was generalized to the setting of a perturbed frame and imprecise knowledge of the measurement matrix in \cite[Theorem 3.1]{AldroubiChenPowell}.

\begin{proof}[of Theorem \ref{thRecoveryWithL2NSP}]
We define the vector $\mathbf{v}:=\mathbf{\hat x}-\mathbf{x}\in\ker\mathbf{M}$ and denote by $S_0\subset [p]$ an index set of $s$ largest absolute entries of $\mathbf{\Omega v}$. 
Since $\# S_0^c=p-s$ and $\mathbf{M}\in\RR^{m\times d}$ satisfies the $\ell_2$-stable $\mathbf{\Omega}$-null space property, it follows
\begin{equation}\label{eqEstimateBySNSP}
\norm{(\mathbf{\Omega v})_{S_0}}_2\leq\frac{\rho}{\sqrt{s}}\norm{\mathbf{\Omega}_{S_0^c} \mathbf{v}}_1\leq\frac{\rho}{\sqrt{s}}\norm{\mathbf{\Omega v}}_1.
\end{equation}
We partition the indices of $S_0^c$ into subsets $S_1$, $S_2$, $\ldots$ of size $s$ in order of decreasing magnitude of $(\Omega v)_i$. Then for each $k\in S_{i+1}$, $i\geq 0$, 
\[
\abs{(\mathbf{\Omega v})_k}\leq\frac{1}{s}\sum_{j\in S_i}\abs{(\mathbf{\Omega v})_j}
\qquad \mbox{ and } \qquad
\norm{(\mathbf{\Omega v})_{S_{i+1}}}_2\leq\frac{1}{\sqrt s}\norm{(\mathbf{\Omega v})_{S_i}}_1.
\]
Along with the triangle inequality this gives
\begin{equation}\label{eqEstimateByDescendingIndexes}
\norm{(\mathbf{\Omega v})_{S_0^c}}_2\leq\sum_{i\geq 1}\norm{(\mathbf{\Omega v})_{S_i}}_2\leq\frac{1}{\sqrt s}\sum_{i\geq 0}\norm{(\mathbf{\Omega v})_{S_i}}_1=\frac{1}{\sqrt s}\norm{\mathbf{\Omega v}}_1.
\end{equation}
Inequalities (\ref{eqEstimateBySNSP}) and (\ref{eqEstimateByDescendingIndexes}) together with Remark~\ref{rem:l1l2} and Theorem~\ref{thConeCostraint} 
yield
\begin{equation}\label{eqEstimateOfL2NormBySNSP}
\norm{ \mathbf{\Omega v}}_2\leq\norm{(\mathbf{\Omega v})_{S_0}}_2+\norm{(\mathbf{\Omega v})_{S_0^c}}_2 \leq
\frac{1+\rho}{\sqrt s}\norm{ \mathbf{\Omega v}}_1 \leq \frac{2(1+\rho)^2}{(1-\rho)\sqrt s} \sigma_s(\mathbf{\Omega x})_1.
\end{equation}
Finally, we use that $\mathbf{\Omega}$ is a frame with lower frame bound $A$ to conclude that
\[
\norm{\mathbf{x}-\mathbf{\hat{x}}}_2\leq\frac{1}{\sqrt A}\norm{\mathbf{\Omega x}-\mathbf{\Omega\hat{x}}}_2 \leq 
\frac{2(1+\rho)^2}{\sqrt{A}(1-\rho) \sqrt{s}} \sigma_s(\mathbf{\Omega x})_1.
\]
This completes the proof. 
\end{proof}

%%%%%%%%%%%%%%%Robust recovery%%%%%%%%%%%%%%%%%%%%

When the measurements are given with some error, the author in \cite{Foucart} introduced the following extension of the $\mathbf{\Omega}$-null space property in order to guarantee robustness of the recovery.
\begin{definition}\label{defL2RobustStableNSP}
A matrix $\mathbf{M}\in\RR^{m\times d}$ is said to satisfy the robust $\ell_2$-stable $\mathbf{\Omega}$-null space property of order $s$ with constant $0<\rho<1$ and $\tau>0$, if for any set $\Lambda\subset [p]$ with $\# \Lambda\geq p-s$ it holds
\begin{equation}\label{eqRobustL2NSP}
\norm{\mathbf{\Omega}_{\Lambda^c}\mathbf{v}}_2\leq\frac{\rho}{\sqrt{s}}\norm{\mathbf{\Omega}_{\Lambda}\mathbf{v}}_1+\tau\norm{\mathbf{Mv}}_2\;\;\;\mbox{for all}\;\;\mathbf{v}\in\RR^d.
\end{equation}
\end{definition}
If $\mathbf v\in\ker\mathbf M$, the term $\norm{\mathbf{Mv}}_2$ vanishes, and we see that the robust $\ell_2$-stable $\mathbf{\Omega}$-null space property implies 
the $\ell_2$-stable $\mathbf{\Omega}$-null space property. 
The robust $\ell_2$-stable null space property guarantees the stability and robustness of the $\ell_1$-minimization (\ref{eqProblemP1Noise}). %Theorem 5 in \cite{Foucart} implies the following $\ell_2$-error estimate.
\begin{theorem}\label{thRecoveryWithRobustL2NSP}
Let $\mathbf{M}\in\RR^{m\times d}$ satisfy the robust $\ell_2$-stable $\mathbf{\Omega}$-null space property of order $s$ with constants $0<\rho<1$ and $\tau>0$. Then for any $\mathbf{x}\in\RR^d$ the solution $\mathbf{\hat x}$ of (\ref{eqProblemP1Noise}) with $\mathbf{y}=\mathbf{M}\mathbf{x}+\mathbf{w}$, $\norm{\mathbf w}_2\leq\eta$, approximates the vector $\mathbf{x}$ with $\ell_2$-error
\begin{equation}\label{eqRobustL2StableRecovery}
\norm{\mathbf{x}-\mathbf{\hat{x}}}_2\leq\frac{2(1+\rho)^2}{\sqrt{A}(1-\rho)}\frac{\sigma_{s}(\mathbf{\Omega x})_1}{\sqrt{s}}+\frac{2\tau(3+\rho)}{\sqrt A(1-\rho)}\eta.
\end{equation}
\end{theorem}
\begin{proof}
Theorem 5 in \cite{Foucart} with $q=p=2$ provides the bound for $\norm{\mathbf\Omega(\mathbf x-\mathbf{\hat x})}_2$. Taking into account that $\mathbf\Omega$ is a frame with lower frame constant $A$, we obtain estimate (\ref{eqRobustL2StableRecovery}).

\end{proof}
%%%%%%%%%%%%%%%%%%%%%%%%%%%%%%%%%%%%%%%%%%%

%%%%%%%%%%%%%%%%%%%%%%%%%%%%%%%%%%%%%%%%%%%%%%%%%

\subsection{Uniform recovery from Gaussian measurements}

We now show Theorem~\ref{thUniformRecoveryWithFrame} by establishing the $\ell_2$-stable $\mathbf{\Omega}$-null space property of order $s$ for a Gaussian measurement
matrix $\mathbf{M}$ by following a similar strategy as in Section~\ref{sec:nonuniform}. To this end we introduce the set 
\[
W_{\rho,s}:=\left\{\mathbf{w}\in\RR^d: \norm{\mathbf{\Omega}_{\Lambda^c}\mathbf{w}}_2>\rho/\sqrt s\norm{\mathbf{\Omega}_{\Lambda}\mathbf{w}}_1\;\mbox{for some}\;\Lambda\subset[p],\; \#\Lambda=p-s\right\}.
\]
In fact, if
\begin{equation}\label{eqMinNormPositivity}
\inf\fbrac{\norm{\mathbf{Mw}}_2:\mathbf{w}\in W_{\rho,s}\cap \mathbb{S}^{d-1}}>0,
\end{equation}
then for all $\mathbf{w}\in\ker\mathbf{M}$ and any $\Lambda\subset [p]$ with $\#\Lambda=p-s$ we have
\[
\norm{\mathbf{\Omega}_{\Lambda^c}\mathbf{w}}_2\leq\frac{\rho}{\sqrt s}\norm{\mathbf{\Omega}_{\Lambda}\mathbf{w}}_1,
\]
which means that $\mathbf M$ satisfies the $\ell_2$-stable $\mathbf{\Omega}$-null space property of order $s$. To show (\ref{eqMinNormPositivity}) we apply Theorem \ref{thModifiedGordonsEscapeThroughTheMesh}, 
which requires to study the Gaussian width of the set $\mathbf{\Omega}\brac{W_{\rho,s}\cap \mathbb{S}^{d-1}}$. Since $\mathbf{\Omega}$ is a frame with upper frame bound $B$, we have
\begin{equation}\label{eqInclusionDueToFrame}
\mathbf{\Omega}\brac{W_{\rho,s}\cap \mathbb{S}^{d-1}}\subset\mathbf{\Omega}\brac{W_{\rho,s}}\cap\brac{\sqrt B \mathbb{B}_2^p}\subset T_{\rho,s}\cap\brac{\sqrt B \mathbb{B}_2^p}=\sqrt B\brac{T_{\rho,s}\cap \mathbb{B}_2^p},
\end{equation}
with
\[
T_{\rho,s}=\left\{\mathbf{u}\in\RR^p: \norm{\mathbf{u}_S}_2\geq\rho/\sqrt s\norm{\mathbf{u}_{S^c}}_1\;\mbox{for some}\;S\subset[p],\; \# S=s\right\}.
\]
Then
\[
T_{\rho,s}\cap \mathbb{B}_2^p=\bigcup\limits_{\#S=s}\left\{ \mathbf{u}\in\RR^p: \norm{\mathbf{u}}_2\leq 1,\;\norm{\mathbf{u}_S}_2>\frac{\rho}{\sqrt s}\norm{\mathbf{u}_{S^c}}_1\right\}.
\]
\begin{lemma}\label{lmInclusionInUniversalSet}
Let $D$ be the set defined by
\begin{equation}\label{eqDefinitionOfD}
D:=\conv\fbrac{\mathbf{x}\in \mathbb S^{p-1}:\#\supp \mathbf{x}\leq s}.
\end{equation}
\begin{enumerate}[(a)]
\item\label{itUnitBall}  Then $D$ is the unit ball with respect to the norm
\[
\norm{\mathbf{x}}_D:=\sum_{l=1}^L\sbrac{\sum_{i\in I_l}\brac{x_i^*}^2}^{1/2},
\]
where $L=\lceil\frac{p}{s}\rceil$, 
%\begin{equation}\label{eq:PartitionOfIndexSet}
\[
I_l=\left\{\begin{array}{ll}
\fbrac{s(l-1)+1,\ldots,sl}, & l=1,\ldots, L-1,\\
\fbrac{s(L-1)+1,\ldots, p}, & l=L,
\end{array}\right.
\]
%\end{equation}
and $\mathbf{x^*}$ is the non-increasing rearrangement of $\mathbf{x}$.
\item\label{itInclusion} It holds
\begin{equation}\label{eqInclusionInUniversalSet}
T_{\rho,s}\cap \mathbb{B}_2^p\subset \sqrt{1+(1+\rho^{-1})^2}D.
\end{equation}
\end{enumerate}
\end{lemma}
A similar result was stated as Lemma 4.5 in \cite{RudelsonVershynin}. For the sake of completeness we present the proof.
\begin{proof}
\ref{itUnitBall} Suppose $\mathbf{x}\in D$. It can be represented as $\mathbf x=\sum\limits_{i}\alpha_i\mathbf x_i$ with $\alpha_i\geq0$, $\sum\limits_i\alpha_i=1$ and $\mathbf x_i\in S^{p-1}$, $\#\supp\mathbf x_i\leq s$. Then $\norm{\mathbf x_i}_D=\norm{\mathbf x_i}_2=1$. By the triangle inequality
\[
\norm{\mathbf x}_D\leq\sum_i\alpha_i\norm{\mathbf x_i}_D=\sum_i\alpha_i=1.
\]
This proves that $D$ is a subset of the unit ball with respect to the $\norm{\cdot}_D$-norm.

On the other hand, let $\norm{\mathbf x}_D\leq 1$. We partition the index set $[p]$ into subsets $S_1$, $S_2$, \ldots of size $s$ in order of decreasing magnitude of entries $x_k$. Set $\alpha_i=\norm{\mathbf x_{S_i}}_2$. Then $\mathbf x$ can be written as
\[
\mathbf x=\sum_{i:\alpha_i\neq 0}\alpha_i\brac{\frac{1}{\alpha_i}\mathbf x_{S_i}},
\qquad
\mbox{ where }  \qquad
\sum_{i:\alpha_i\neq 0}\alpha_i=\sum\limits_{i}\norm{\mathbf x_{S_i}}_2=\norm{\mathbf x}_D\leq 1
\]
and, for $\alpha_i\neq 0$, $\norm{\frac{1}{\alpha_i}\mathbf x_{S_i}}_2=\frac{1}{\alpha_i}\norm{\mathbf x_{S_i}}_2=1$.
Thus $\mathbf x\in D$.

\ref{itInclusion} Take an arbitrary $\mathbf{x}\in T_{\rho,s}\cap \mathbb{B}_2^p$. To show (\ref{eqInclusionInUniversalSet}) we estimate $\norm{\mathbf x}_D$. According to the definition of $\norm{\cdot}_D$ in Lemma \ref{lmInclusionInUniversalSet} \ref{itUnitBall},
\begin{align}
\norm{\mathbf x}_D&=\sum_{l=1}^L\sbrac{\sum_{i\in I_l}\brac{x_i^*}^2}^{\frac{1}{2}}\notag\\
&=\sbrac{\sum_{i=1}^s\brac{x_i^*}^2}^{\frac{1}{2}}+\sbrac{\sum_{i=s+1}^{2s}\brac{x_i^*}^2}^{\frac{1}{2}}+\sum_{l\geq 3}^L\sbrac{\sum_{i\in I_l}\brac{x_i^*}^2}^{\frac{1}{2}}\label{eq:EstimateForDNormInBlocks}.
\end{align}
%\begin{align}
%\norm{\mathbf x}_D&=\sum_{l=1}^L\sbrac{\sum_{i\in I_l}\brac{x_i^*}^2}^{\frac{1}{2}}=\sbrac{\sum_{i=1}^s\brac{x_i^*}^2}^{\frac{1}{2}}+\sbrac{\sum_{i=s+1}^{2s}\brac{x_i^*}^2}^{\frac{1}{2}}+\sum_{l\geq 3}^L\sbrac{\sum_{i\in I_l}\brac{x_i^*}^2}^{\frac{1}{2}\notag}\\
%&\leq \sqrt2\sbrac{\sum_{i=1}^{2s}(x^*_i)^2}^{\frac{1}{2}}+\sum_{l\geq 3}^L\sbrac{\sum_{i\in I_l}\brac{x_i^*}^2}^{\frac{1}{2}}\leq \sqrt 2+\sum_{l\geq 3}^L\sbrac{\sum_{i\in I_l}\brac{x_i^*}^2}^{\frac{1}{2}}\label{eq:EstimateForDNormInBlocks}.
%\end{align}
To bound the last term in the inequality above, we first note that for each $i\in I_{l}$, $l\geq 3$,
\[
x^*_i\leq\frac{1}{s}\sum_{j\in I_{l-1}}x^*_j\quad\text{and}\quad \sbrac{\sum_{i\in I_{l}}(x^*_i)^2}^{1/2}\leq\frac{1}{\sqrt s}\sum_{j\in I_{l-1}}x^*_j.
\]
Summing up over $l\geq 3$ yields
\[
\sum_{l\geq 3}^L\sbrac{\sum_{i\in I_l}\brac{x_i^*}^2}^{\frac{1}{2}}\leq\frac{1}{\sqrt s}\sum_{l\geq 2}\sum_{j\in I_l}x^*_j.
\]
Since $\mathbf{x}\in T_{\rho,s}\cap \mathbb{B}_2^p$, it holds $\norm{\mathbf{x}}_2\leq 1$ and there is $S\subset [p]$, $\# S=s$, such that $\norm{\mathbf{x}_S}_2>\rho/\sqrt s\norm{\mathbf{x}_{S^c}}_1$. Then
\[
\sum_{l\geq 2}\sum_{i\in I_l}x^*_i\leq\norm{\mathbf x_{S^c}}_1<\frac{\sqrt s}{\rho}\norm{\mathbf x_S}_2\leq\frac{\sqrt s}{\rho}\sbrac{\sum_{i=1}^s(x_i^*)^2}^{\frac{1}{2}}
\]
and
\[
\sum_{l\geq 3}^L\sbrac{\sum_{i\in I_l}\brac{x_i^*}^2}^{\frac{1}{2}}\leq\rho^{-1}\sbrac{\sum_{i=1}^s(x_i^*)^2}^{\frac{1}{2}}.
\]
Applying the last estimate to (\ref{eq:EstimateForDNormInBlocks}) and taking into account that $\norm{x}_2^2\leq 1$ we derive that
\[
%\norm{x}_D\leq \sqrt 2+\rho^{-1}
\begin{aligned}
\norm{\mathbf x}_D&\leq(1+\rho^{-1})\sbrac{\sum_{i=1}^s(x_i^*)^2}^{\frac{1}{2}}+\sbrac{\sum_{i=s+1}^{2s}\brac{x_i^*}^2}^{\frac{1}{2}}\\
&\leq1+\rho^{-1})\sbrac{\sum_{i=1}^s(x_i^*)^2}^{\frac{1}{2}}+\sbrac{1-\sum_{i=1}^{s}\brac{x_i^*}^2}^{\frac{1}{2}}.
\end{aligned}
\]
Set $a=\sbrac{\sum_{i=1}^s(x_i^*)^2}^{\frac{1}{2}}$. The maximum of the function
\[
f(a):=(1+\rho^{-1})a+\sqrt{1-a^2}, \quad 0\leq a\leq 1,
\]
is attained at the point
\[
a = \frac{1+\rho^{-1}}{\sqrt{1+(1+\rho^{-1})^2}}
\]
and is equal to $\sqrt{1+(1+\rho^{-1})^2}$. Thus for any $\mathbf{x}\in W$ it holds
\[
\norm{\mathbf x}_D\leq\sqrt{1+(1+\rho^{-1})^2},
\]
which proves (\ref{eqInclusionInUniversalSet}).
\end{proof}
Lemma \ref{lmInclusionInUniversalSet} \ref{itInclusion} implies
\begin{equation}\label{eqGaussianWidthOfConeAndD}
\ell\brac{T_{\rho,s}\cap \mathbb{B}_2^p}\leq \sqrt{1+(1+\rho^{-1})^2}\ell(D).
\end{equation}

\begin{lemma}\label{lmEstimateGaussianWidthOfD}
The Gaussian width of the set $D$ defined by (\ref{eqDefinitionOfD}) satisfies
\[
\ell(D)\leq\sqrt{2s\ln\frac{ep}{s}}+\sqrt s.
\]
\end{lemma}
\begin{proof}
The supremum of the linear functional $\abrac{\mathbf{g},\mathbf{x}}$ over $D$ is achieved at an extreme point, i.e., at an $\mathbf{x}\in S^{p-1}$ with $\#\supp \mathbf{x}\leq s$. Hence, by H\"older's inequality
\[
\ell(D)=\mean\underset{\mathbf{x}\in D}\sup\abrac{\mathbf{g},\mathbf{x}}=\mean\underset{\begin{subarray}{c}
\norm{\mathbf{x}}_2=1,\\
 \#\supp \mathbf{x}\leq s
 \end{subarray}}\sup\abrac{\mathbf{g},\mathbf{x}}\leq\mean\underset{S\subset [p], \#S=s}\max\norm{\mathbf{g}_S}_2.
 \]
An estimate on the maximum squared $\ell_2$-norm of a sequence of standard Gaussian random vectors (see e.g.\ \cite[Lemma 3.2]{RaoRechtNowak} or \cite[Proposition 8.2]{FoucartRauhut}) 
gives
\[
\ell(D)\leq\sqrt{\mean\underset{S\subset [p], \#S=s}\max\norm{\mathbf{g}_S}_2^2}
\leq\sqrt{2\ln\binom{p}{s}}+\sqrt s\leq\sqrt{2s\ln\frac{ep}{s}}+\sqrt s.
\]
The last inequality follows from the fact that
$
\binom{p}{s}\leq\brac{\frac{ep}{s}}^s$, see e.g.~\cite[Lemma C.5]{FoucartRauhut}. 
\end{proof}

\begin{proof}[of Theorem \ref{thUniformRecoveryWithFrame}]
Expressions (\ref{eqInclusionDueToFrame}), (\ref{eqGaussianWidthOfConeAndD}) and Lemma \ref{lmEstimateGaussianWidthOfD} show that
\begin{align}
\ell\brac{\mathbf{\Omega}\brac{W_{\rho,s}\cap \mathbb{S}^{d-1}}}&\leq\sqrt{B\sbrac{1+(1+\rho^{-1})^2}}\ell(D)\notag\\
&\leq\sqrt{B\sbrac{1+(1+\rho^{-1})^2}}\brac{\sqrt{2s\ln\frac{ep}{s}}+\sqrt s}\label{eq:EstimateGWImageOfTheSetUniformRecovery}.
\end{align}
Set $t=\sqrt{2\ln(\eps^{-1})}$. The fact that $E_m\geq m/\sqrt{m+1}$ along with condition (\ref{eqNumberOfMeasurementsForFrameUniformRecovery}) yields 
\[
E_m\geq \frac{1}{\sqrt{A}}\ell \brac{\mathbf{\Omega}\brac{W_{\rho,s}\cap \mathbb{S}^{d-1}}}+t.
\]
The monotonicity of probability and Theorem \ref{thModifiedGordonsEscapeThroughTheMesh} imply
\[
\PP\brac{\inf\norm{\mathbf{Mw}}_2> 0:\mathbf{w}\in W_{\rho,s}\cap \mathbb{S}^{d-1}}\geq 1-e^{-\frac{t^2}{2}}=1-\eps,
\]
which guarantees that with probability at least $1-\eps$
\[
\norm{\mathbf{\Omega}_{\Lambda^c}\mathbf{w}}_2<\frac{\rho}{\sqrt s}\norm{\mathbf{\Omega}_{\Lambda}\mathbf{w}}_1
\]
for all $\mathbf{w}\in\ker\mathbf{M}\setminus\{\mathbf 0\}$ and any $\Lambda\subset [p]$ with $\#\Lambda=p-s$, see (\ref{eqMinNormPositivity}). This means that $\mathbf{M}$ satisfies the $\ell_2$-stable $\mathbf{\Omega}$-null space property of order $s$. Finally, we apply Theorem \ref{thRecoveryWithL2NSP}. 
\end{proof}

%%%%%%%%%%%added part%%%%%%%%%%%%%%%%%%%%%%%%%%%%%%%%
Finally, we extend to robustness of the recovery with respect to perturbations of the measurements. 
\begin{theorem}\label{thRobustUniformRecoveryWithFrame}
Let $\mathbf{M}\in\RR^{m\times d}$ be a Gaussian random matrix, $0<\rho<1$, $0<\eps<1$ and $\tau> 1$. If
\begin{equation}\label{eqNumberOfMeasurementsForFrameRobustUniformRecovery}
\frac{m^2}{m+1}\geq   \frac{2\brac{1+(1+\rho^{-1})^2}\tau^2B}{(\tau-1)^2A}\, s\, \brac{\!\!\sqrt{\ln\frac{ep}{s}}+\frac{1}{\sqrt 2}+\sqrt{\frac{A\ln(\eps^{-1})}{Bs\brac{1+(1+\rho^{-1})^2}}}}^{\!\!2}\!\!, 
%m\geq\brac{1-\frac{1}{\tau}}^{\!-2}\!\!\sbrac{\frac{2Bs\brac{\sqrt 2+\rho^{-1}}^{\!2}}{A}\brac{\!\!\sqrt{\ln\frac{ep}{s}}+\frac{1}{\sqrt 2}+\frac{1}{\sqrt 2+\rho^{-1}}\sqrt{\frac{A\ln(\eps^{-1})}{Bs}}}^{\!\!2}}\!\!, 
\end{equation}
then with probability at least $1-\eps$ for every vector $\mathbf{x}\in\RR^d$ and perturbed measurements $\mathbf{y} = \mathbf{\Omega x} + \mathbf{w}$ with $\|\mathbf{w}\|_2 \leq \eta$ 
a minimizer $\mathbf{\hat x}$ of (\ref{eqProblemP1Noise}) approximates $\mathbf{x}$ with $\ell_2$-error
\[
\norm{\mathbf{x}-\mathbf{\hat{x}}}_2\leq\frac{2(1+\rho)^2}{\sqrt{A}(1-\rho)}\frac{\sigma_{s}(\mathbf{\Omega x})_1}{\sqrt{s}}+\frac{2\tau\sqrt{2B}(3+\rho)}{\sqrt m\sqrt A(1-\rho)}\eta.
\]
\end{theorem}
\begin{proof}
Condition (\ref{eqNumberOfMeasurementsForFrameRobustUniformRecovery}) together with $E_m\geq m/\sqrt{m+1}$ imply
%\[
%E_m^2\geq\frac{m^2}{m+1}
%\sqrt m\brac{\frac{1}{\sqrt 2}-\frac{1}{\tau}}\geq \sqrt{\frac{B}{A}(1+(1+\rho^{-1})^2)}\brac{\sqrt{2s\ln\frac{ep}{s}}+\sqrt s}+\sqrt{2\ln(\eps^{-1})}.
%\]
%Together with 
%\[
%\frac{E_m}{\sqrt m}\geq\sqrt\frac{m}{m+1}\geq\frac{1}{\sqrt 2}, \quad m\geq 1,
%\]
%this implies 
\[
E_m\brac{1-\frac{1}{\tau}}\geq\sqrt{\frac{B}{A}(1+(1+\rho^{-1})^2)}\brac{\sqrt{2s\ln\frac{ep}{s}}+\sqrt s}+\sqrt{2\ln(\eps^{-1})},
%\sqrt m\brac{\frac{E_m}{\sqrt m}-\frac{1}{\tau}}\geq \sqrt{\frac{B}{A}(1+(1+\rho^{-1})^2)}\brac{\sqrt{2s\ln\frac{ep}{s}}+\sqrt s}+\sqrt{2\ln(\eps^{-1})},
\]
%\[
%\brac{\sqrt{m-1}-\frac{\sqrt m}{\tau}}^2\geq\brac{\frac{\sqrt B(\sqrt 2+\rho^{-1})(\sqrt{2s\ln\frac{ep}{s}}+\sqrt s)}{\sqrt A}+\sqrt{2\ln(\eps^{-1})}}^2,
%\]
which is equivalent to
\[
E_m- \sqrt{\frac{B}{A}(1+(1+\rho^{-1})^2)}\brac{\sqrt{2s\ln\frac{ep}{s}}+\sqrt s}-\sqrt{2\ln(\eps^{-1})}\geq\frac{E_m}{\tau}.
\]
%\sqrt{m-1}-\frac{\sqrt B(\sqrt 2+\rho^{-1})(\sqrt{2s\ln\frac{ep}{s}}+\sqrt s)}{\sqrt A}-\sqrt{2\ln(\eps^{-1})}\geq \frac{\sqrt m}{\tau}.
Taking into account (\ref{eq:EstimateGWImageOfTheSetUniformRecovery}) we may conclude
\[
E_m-\frac{1}{\sqrt A}\ell\brac{\mathbf\Omega\brac{W_{\rho,s}\cap\mathbb S^{d-1}}}-\sqrt{2\ln(\eps^{-1})}\geq\frac{E_m}{\tau}\geq\frac{1}{\tau}\sqrt\frac{m}{2}.
\]
%Note that $E_m\geq m/\sqrt{m+1}>\sqrt{m-1}$ and (\ref{eq:EstimateGWImageOfTheSetUniformRecovery}) imply
%\[
%\begin{aligned}
%&E_m-\frac{1}{\sqrt A}\ell\brac{\mathbf\Omega\brac{W_{\rho,s}\cap\mathbb S^{d-1}}}-\sqrt{2\ln(\eps^{-1})}\\
%>&\sqrt{m-1}-\frac{\sqrt B(\sqrt 2+\rho^{-1})(\sqrt{2s\ln\frac{ep}{s}}+\sqrt s)}{\sqrt A}-\sqrt{2\ln(\eps^{-1})}\geq \frac{\sqrt m}{\tau}.
%\end{aligned}
%\]
Then according to Theorem \ref{thModifiedGordonsEscapeThroughTheMesh} 
\[
\PP\brac{\inf\norm{\mathbf{Mw}}_2> \frac{\sqrt m}{\tau\sqrt 2}:\mathbf{w}\in W_{\rho,s}\cap \mathbb{S}^{d-1}} \geq 1-\eps.
\]
This means that for any $\mathbf w\in\RR^d$ such that $\norm{\mathbf{Mw}}_2\leq\frac{\sqrt m}{\tau\sqrt 2}\norm{\mathbf w}_2$ and any set $\Lambda\subset [p]$ with $\# \Lambda\geq p-s$ it holds with probability at least $1-\eps$
\[
\norm{\mathbf{\Omega}_{\Lambda^c}\mathbf{w}}_2<\frac{\rho}{\sqrt{s}}\norm{\mathbf{\Omega}_{\Lambda}\mathbf{w}}_1.
\]
For the remaining vectors $\mathbf w\in\RR^d$, we have $\norm{\mathbf{Mw}}_2>\frac{\sqrt m}{\tau\sqrt 2}\norm{\mathbf w}_2$, which together with the fact that $\mathbf\Omega$ is a frame with upper frame bound $B$ leads to
\[
\norm{\mathbf{\Omega}_{\Lambda^c}\mathbf{w}}_2\leq\norm{\mathbf{\Omega w}}_2\leq\sqrt B\norm{\mathbf w}_2<\frac{\tau\sqrt {2B}}{\sqrt m}\norm{\mathbf{Mw}}_2.
\]
Thus, for any $\mathbf w\in\RR^d$,
\[
\norm{\mathbf{\Omega}_{\Lambda^c}\mathbf{w}}_2<\frac{\rho}{\sqrt{s}}\norm{\mathbf{\Omega}_{\Lambda}\mathbf{w}}_1+\frac{\tau\sqrt{2B}}{\sqrt m}\norm{\mathbf{Mw}}_2.
\]
Finally, we apply Theorem \ref{thRecoveryWithRobustL2NSP}. 
\end{proof}

%%%%%%%%%%%%%%%%%%%%%%%%%%%%%%%%%%%%%%%%%%%%%%%%%

\section{Numerical experiments}

In this section we present the results of numerical experiments on synthetic data performed in Matlab using the \texttt{cvx} package. 

For the first set of experiments we constructed tight frames $\mathbf{\Omega}$ as an orthonormal basis of the range of the matrix the rows of which were drawn randomly and independently from $\mathbb{S}^{d-1}$. 
In order to obtain also non-tight frames we simply varied the norms of the rows of $\mathbf{\Omega}$. As dimensions for the analysis operator, we have chosen $d=200$ and $p=250$. The maximal number of zeros that can be achieved in the analysis representation 
$\mathbf{\Omega x}$ is less than $d$, since otherwise $\mathbf{x}=0$. Therefore, the sparsity level of $\mathbf{\Omega x}$ was always greater than $50$. For each trial we fixed a cosparsity $l$ (resulting
in the sparsity $s=p-l$)
and selected at random $l$ rows of an analysis operator $\mathbf{\Omega}$ that constitute the cosupport $\Lambda$ of the signal. To produce a signal $\mathbf{x}$ we constructed 
a basis $\mathbf{B}$ of $\ker \mathbf{\Omega}_{\Lambda}$, drew a coefficient vector $\mathbf{c}$ from a normalized standard Gaussian distribution and set $\mathbf{x}=\mathbf{Bc}$. 
We ran the algorithm and counted the number of times the signal was recovered correctly out of $70$ trials. A reconstruction error of less than $10^{-5}$ was considered as a successful recovery. 
The curves in Figure~\ref{figDifFramesMeasurementsVsSparsity} depict the relation between the number of measurements and the sparsity level such that the recovery was successful at least $98\%$ of the time.  
Each point on the line corresponds to the maximal sparsity level that could be achieved for the given number of measurements.    
\begin{figure}[hbt]
\centering
\includegraphics[scale=0.5 ]{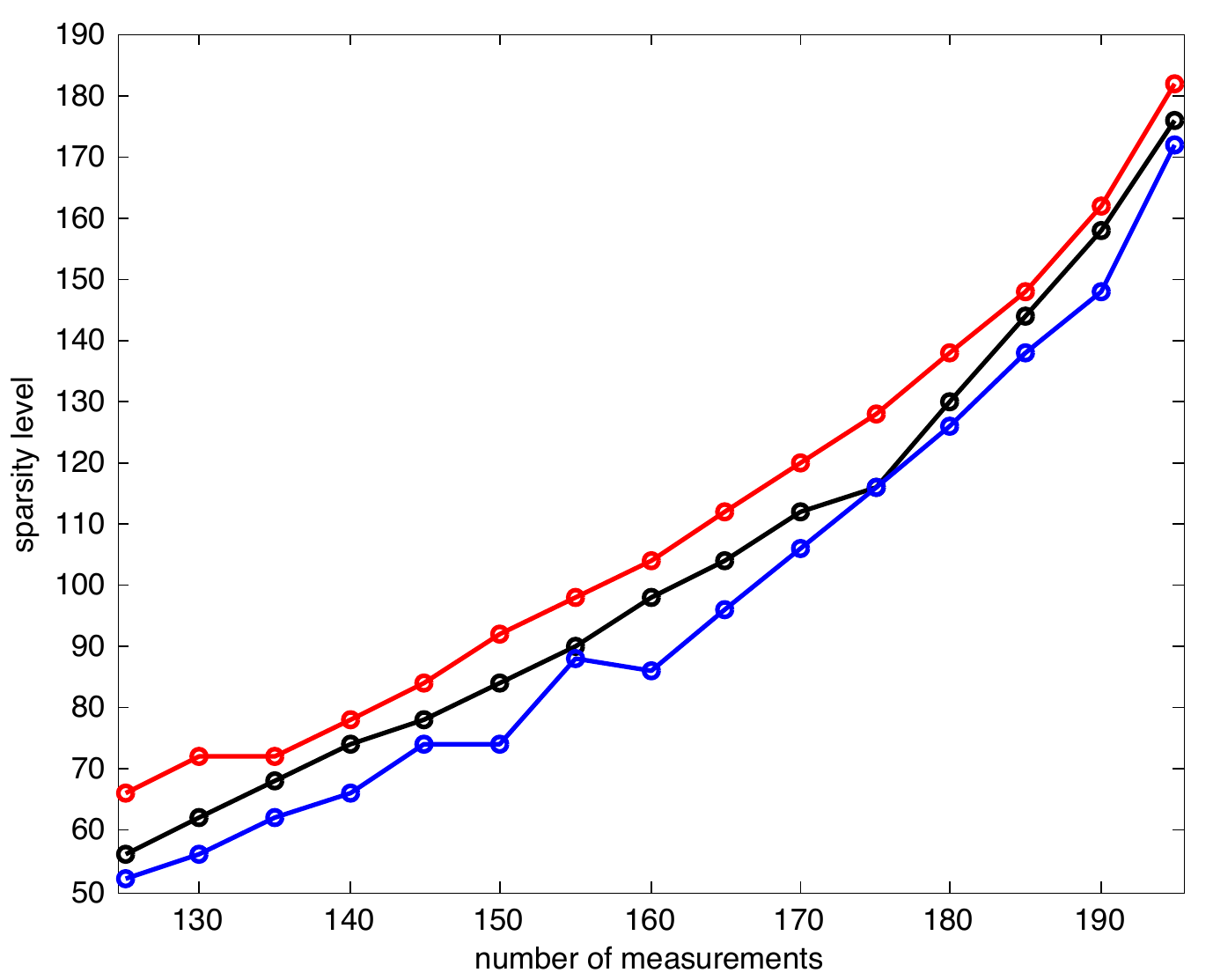}
\caption{Recovery for different analysis operators. The red curve corresponds to a tight frame, the black one has frame bound ratio $B/A$ of $13.1254$ and for the blue one $B/A=45.7716$.}\label{figDifFramesMeasurementsVsSparsity}
\end{figure}

%For the second set of experiments we use Gabor frames as analysis operators $\mathbf{\Omega}$. We get different frame ratios $B/A$ by changing the width of a Gaussian window. We take $d=200$, $p=250$ and set number of trials to $50$. The $98\%$ recovery trend is displayed in Figure \ref{figMeasurementsVsSparsityGaborFrames}.
%\begin{figure}[ht]
%\centering
%\includegraphics[scale=0.4 ]{Measurements_Vs_Sparsity_Gabor_Frames}
%\caption{The red curve corresponds to the frame with $B/A=27.8611$ and for the blue one $B/A=2.6806$.}\label{figMeasurementsVsSparsityGaborFrames}
%\end{figure}
%
%Finally we compare the performance of two tight Gabor frames with different windows. The dimension of a signal is $d=200$ and the number of elements in each frame is $p=250$. Number of trials is $50$.
%\begin{figure}[ht]
%\centering
%\includegraphics[scale=0.4 ]{Comparison_Tight_Gabor_Frames}
%\caption{Comparison of two tight Gabor frames with gaussian windows of different size.}\label{figMeasurementsVsSparsityTightGaborFrames}
%\end{figure}

The experiments clearly show that analysis $\ell_1$-minimization works very well for recovering cosparse signals from Gaussian measurements.
(Note that only a comparison of the experiments with the nonuniform recovery guarantees make sense.) The frame bound ratio $B/A$ indeed influences the performance of the recovery algorithm (\ref{eqProblemP1}) --
although the degradation with increasing value of $B/A$ is less dramatic than indicated by our theorems. The reason for this may be that the theorems give estimates for the worst case, while
the experiments can only reflect the typical behavior.

%%%%%%%%%%%%%%%%%%%%%%%%%%%%%%%%%%%%%%%%%%%%%%%%%

\section*{Acknowledgements}
M.~Kabanava and H.~Rauhut acknowledge support by the Hausdorff Center for Mathematics, University of Bonn, and by the European Research Council through
the grant StG 258926.

%%%%%%%%%%%%%%%%%%%%%%%%%%%%%%%%%%%%%%%%%%%%%%%%%


%%%%%%%%%%%%%%%%%%%%%%%%%%%%%%%%%%%%%%%%%%%%%%%%%

%%%%%%%%%%%%%%%%%%%%%%%%%%%%%%%%%%%%%%%%%%%%%%%%

\end{document}